\documentclass[a4paper,USenglish,cleveref, autoref]{lipics-v2019}

\usepackage{stmaryrd}

\title{Compiling Petri Net Mutual Reachability in Presburger} 

\titlerunning{Petri Net Mutual Reachability Relations}

\author{Jérôme Leroux}{Univ. Bordeaux, CNRS, Bordeaux INP, LaBRI, UMR 5800, F-33400 Talence, France}{jerome.leroux@labri.fr}{}{The author is supported by the grant ANR-17-CE40-0028 of the French National Research Agency ANR (project BRAVAS)}

\authorrunning{J. Leroux}

\Copyright{Jérôme Leroux}



\category{}

\supplement{}

\nolinenumbers 

\hideLIPIcs  

\EventEditors{John Q. Open and Joan R. Access}
\EventNoEds{2}
\EventLongTitle{42nd Conference on Very Important Topics (CVIT 2016)}
\EventShortTitle{CVIT 2016}
\EventAcronym{CVIT}
\EventYear{2016}
\EventDate{December 24--27, 2016}
\EventLocation{Little Whinging, United Kingdom}
\EventLogo{}
\SeriesVolume{42}
\ArticleNo{23}

\newcommand{\setN}{\mathbb{N}}
\newcommand{\setZ}{\mathbb{Z}}
\newcommand{\setQ}{\mathbb{Q}}

\newcommand{\eqby}[1]{\stackrel{\textrm{{\normalfont\tiny{#1}}}}{=}}
\newcommand{\eqdef}{\eqby{def}}

\newcommand{\src}[1]{\operatorname{src}(#1)}
\newcommand{\tgt}[1]{\operatorname{tgt}(#1)}

\newcommand{\norm}[1]{{\mathopen{\|}#1\mathclose{\|}}}
\newcommand{\normi}[1]{{\mathopen{\|}#1\mathclose{\|}_\infty}}
\newcommand{\crochet}[1]{\llbracket{#1}\rrbracket}

\renewcommand{\vec}[1]{\mathbf{#1}}

\makeatletter
\newcommand\xleftrightarrow[2][]{%
  \ext@arrow 9999{\longleftrightarrowfill@}{#1}{#2}}
\newcommand\longleftrightarrowfill@{%
  \arrowfill@\leftarrow\relbar\rightarrow}
\makeatother

\begin{document}

\maketitle

\begin{abstract}
  Petri nets are a classical model of concurrency widely
  used and studied in formal verification with many
  applications in modeling and analyzing hardware and software, data
  bases, and reactive systems. The reachability problem is central
  since many other problems reduce to reachability questions. The reachability problem is known to be decidable but its complexity is extremely high (non primitive recursive).
  In 2011, a variant of the reachability problem, called
  the mutual reachability problem, that consists in deciding if two configurations are mutually reachable was proved to be exponential-space
  complete. Recently, this problem found several
  unexpected applications in particular in the theory of
  population protocols. While the mutual reachability problem is known to be definable in the Preburger arithmetic, the best known upper bound of such a formula was recently proved to be non-elementary (tower). In this paper we provide a way to compile the mutual reachability relation of a Petri net with $d$ counters into a quantifier-free Presburger formula given as a doubly exponential disjunction of $O(d)$ linear constraints of exponential size. We also provide some first results about Presburger formulas encoding bottom configurations.
\end{abstract}

\section{Introduction}
Petri nets are a classical model of concurrency widely
used and studied in formal verification with many
applications in modeling and analyzing hardware and software, data
bases, and reactive systems. The reachability problem is central
since many other problems reduce to reachability
questions. Unfortunately, the reachability problem is difficult for
several reasons. In fact, from a complexity point of view,
the problem was recently proved to be Ackermannian-complete~(\cite{lics19} for the upper-bound and \cite{DBLP:conf/focs/CzerwinskiO21,DBLP:conf/focs/Leroux21} for equivalent lower-bounds).
Moreover,
even in practice, the reachability problem is difficult. Nowadays, no
efficient tool exists for deciding it since the known algorithms are
difficult to be implemented and require many enumerations in
extremely large state spaces (see \cite{DBLP:conf/tacas/DixonL20} for the
state-of-the-art algorithm deciding the general reachability problem).

\medskip

Fortunately, easier natural variants of the reachability problem can be applied
in various contexts. For instance, the coverability problem, a variant of the reachability problem, can be applied in the analysis of concurrent
programs~\cite{DBLP:conf/cav/BaslerMWK09}. The coverability problem is known to be
exponential-space complete~\cite{Rackoff78,Lipton76}, and efficient
tools exist~\cite{DBLP:conf/tacas/BlondinFHH16,
  DBLP:journals/tcs/GeffroyLS18}.

\medskip

Another variant is the mutual
reachability problem. This problem consists in deciding if two
configurations are mutually reachable one from the other. This problem
was proved to be exponential-space complete
in~\cite{concurjournal13} and finds unexpected applications in
population protocols~\cite{DBLP:journals/acta/EsparzaGLM17}, trace
logics~\cite{DBLP:conf/concur/LerouxPS13}, universality problems
related to structural liveness
problems~\cite{DBLP:conf/apn/JancarLS18}, and in solving the home state
problem \cite{DBLP:journals/ipl/BestE16}.
The exponential-space complexity upper-bound of the mutual
reachability problem proved in \cite{concurjournal13} is obtained by
observing that if two configurations are mutually reachable, then
the two configurations
belong to a cycle of the (infinite) reachability graph with a length
at most doubly-exponential with respect to the size in binary of the
two configurations.

\medskip

Recently, the computation of a Presburger formula encoding the mutual reachability problem found an application in the reachability problem of Petri nets extended with a stack in~\cite{DBLP:conf/icalp/GanardiMPSZ22}. In that paper, authors provided a formula of tower size and left open the computation of a formula of elementary size.

\medskip

\textbf{Contribution}.
In this paper, we focus on the minimal size of a quantifier-free Presburger formula encoding the mutual reachability relation. We provide a way to compile the mutual reachability relation of a Petri net with $d$ counters into a quantifier-free Presburger formula given as a doubly exponential disjunction of $O(d)$ linear constraints of exponential size.

\medskip


\textbf{Outline}.
In \cref{part:compiling} we provide algorithms for encoding in the quantifier-free fragment of the Presburger arithmetic the membership of vectors in lattices (see~\cref{sec:lattice}), and the reachability relations of words of Petri net actions (see~\cref{sec:action}). In \cref{part:overview} we provide an overview of the way quantifier-free Presburger formula encoding the mutual reachability are obtained. Results of that part are self-contains except two technical results that are proved respectively in \cref{part:lem1} and \cref{part:lem2}. In \cref{part:conclusion} we conclude the paper with an open problem about the configurations of a Petri net in the bottom strongly-connected components of the reachability graph.

\clearpage

\part{Compiling in Presburger}\label{part:compiling}
In this part, we provide algorithms for encoding in the quantifier-free fragment of the Presburger arithmetic lattices (see~\cref{sec:lattice}) and reachability relations of words of Petri net actions (see~\cref{sec:action}).

\medskip

The set of rationals, integers, non-negative rationals, and natural numbers are denoted as $\setQ$, $\setZ$, $\setQ_{\geq 0}$, and $\setN$ respectively. We denote by $\mathbb{X}^d$ the set of $d$-dimensional vectors of elements in $\mathbb{X}$. Vectors are denoted in bold face, and given $\vec{x}\in\setQ^d$, we denote by $\vec{x}(1),\ldots,\vec{x}(d)$ its components in such a way $\vec{x}=(\vec{x}(1),\ldots,\vec{x}(d))$. We denote by $\norm{\vec{x}}$ the \emph{one-norm} $\sum_{i=1}^d|\vec{x}(i)|$, and by $\normi{\vec{x}}$ the \emph{infinite-norm} $\max_{1\leq i\leq d}|\vec{x}(i)|$. Given two vectors $\vec{x},\vec{y}\in\mathbb{Q}^d$, we denote by $\vec{x}\cdot\vec{y}$ the number $\sum_{i=1}^d\vec{x}(i)\vec{y}(i)$.

\section{Lattices}\label{sec:lattice}
A \emph{lattice} is a \emph{subgroup} of $(\setZ^d,+)$, i.e. a set $\vec{L}\subseteq \setZ^d$ that contains the zero vector $\vec{0}$, such that $\vec{x}+\vec{y}\in \vec{L}$ for every $\vec{x},\vec{y}\in\vec{L}$, and such that $-\vec{x}\in\vec{L}$ for every $\vec{x}\in\vec{L}$. The lattice \emph{spanned} by a finite sequence $\vec{v}_1,\ldots,\vec{v}_k$ of vectors in $\setZ^d$ is the lattice $\setZ\vec{v}_1+\cdots+\setZ\vec{v}_k$. Let us recall that any lattice is spanned by such a finite sequence. It follows that for any lattice $\vec{L}\subseteq \setZ^d$, there is a Presburger formula encoding the membership of a vector $\vec{v}$ in $\vec{L}$ of the form $\exists z_1,\ldots,z_k~\vec{v}=z_1\vec{v}_1+\cdots+z_k\vec{v}_k$. In order to obtain a quantifier-free Presburger formula encoding $\vec{L}$, we can perform a quantifier elimination from such a formula but it is difficult to obtain interesting complexity bounds with such an approach. We follow another approach based on the notion of \emph{representations}.

\medskip

A \emph{representation} $\gamma$ of a lattice $\vec{L}\subseteq \setZ^d$ is a tuples $\gamma\eqdef ((n_1,\vec{a}_1),\ldots,(n_d,\vec{a}_d))$ of $d$ pairs $(n_i,\vec{a}_i)\in\setN\times\setZ^d$ such that the following equality holds:
$$\vec{L}=\left\{\vec{x}\in \setZ^d \mid \bigwedge_{i=1}^d \vec{a}_i\cdot \vec{x}\in n_i\setZ\right\}$$
We denote by $\crochet{\gamma}$ the lattice $\vec{L}$ and by $\norm{\gamma}_\infty$ the value $\max\{n_1,\norm{\vec{a}_1}_\infty,\ldots,n_d,\norm{\vec{a}_d}_\infty\}$.

\medskip

In this section we prove the following~\cref{thm:lattice} that shows that lattices spanned by finite sequences of vectors admits a representation computable in polynomial time (assuming numbers encoded in binary). Such a representation will help encoding lattice membership with a formula in the quantifier-free fragment of the Presburger arithmetic. All other results and definitions are not used in the sequel.
\begin{theorem}\label{thm:lattice}
  Let $\vec{L}$ be the lattice spanned by a sequence of vectors in $\{-m,\ldots,m\}^d$ encoded in binary for some $m\in\setN$. We can compute in polynomial time a representation $\gamma$ of $\vec{L}$ such that $\norm{\gamma}_\infty\leq (d!)^2m^d$.
\end{theorem}

The proof is based on classical matrix operations over the rationals. We do not recall classical notations and definitions we are using but briefly recall some others.

\medskip
\newcommand{\com}{\operatorname{com}}

Let $M$ be a $r\times r$ square matrix. Let us recall that the \emph{determinant} of $M$, denoted as $\det(M)$ is defined as $\sum_{\sigma\in S_r}\operatorname{sgn}(\sigma)\prod_{i=1}^r M_{i\sigma(i)}$ where $S_r$ is the set of permutations of $\{1,\ldots,r\}$, and $\operatorname{sgn}(\sigma)\in\{-1,1\}$ is the sign of $\sigma$. In particular if $M$ is an integer matrix with coefficients bounded by some $m\in\setN$, then $\det(M)$ is an integer bounded in absolute value by $r!m^r$. A $r\times r$ square matrix is said to be \emph{non-singular} if $\det(M)$ is nonzero. Notice that in this case there exists a unique $r\times r$ matrix $M^{-1}$ such that $M^{-1}M=MM^{-1}=I_r$ where $I_r$ is the identity $r\times r$ square matrix, i.e. the matrix with zero coefficients except on the main diagonal where coefficients are one. This matrix can be computed by introducing the comatrice. Let us recall that the comatrix of any $r\times r$ square matrix $M$ is the $r\times r$ square rational matrix $\com(M)$ satisfying $(\com(M))_{i j}$ is $\det(M^{i j})$ where $M^{i j}$ is the matrix obtained from $M$ by replacing the $j$th column by a column of zeros, except on the $i$th line in which we put a one. Let us recall that $M^\top\com(M)=\det(M)I_r$ where $M^\top$ is the transpose of $M$. It follows that if $M$ is non-singular then $M^{-1}=\frac{1}{\det(M)}\com(M)^\top$. Notice that if the coefficients of $M$ are integers bounded in absolute value by some $m\in\setN$, then the coefficients of $\operatorname{com}(M)$ are integers bounded in absolute value by $(r-1)!m^{r-1}$.

\medskip

Let $M$ be a $r\times k$ matrix. The rank of $M$ is the maximal $n\in\{0,\ldots,\min\{r,k\}\}$ such that $M$ admits a $n\times n$ non-singular sub-matrix. Such a matrix $M$ is said to be \emph{full row rank} if its rank is equal to $r$. We say that such a matrix $M$ has an \emph{Hermite normal form} if there exists a uni-modular matrix $U$ (i.e. a matrix of integers with a $+1$ or $-1$ determinant) such that $MU=[H 0]$ where $H$ is a non-singular, lower triangular, non-negative matrix, in which each row has a unique maximum entry, which is located on the main diagonal of $H$. Let us recall that every full row rank matrix $M$ has a unique Hermite normal form $[H0]$ and moreover the uni-modular matrix $U$ such that $MU=[H 0]$ is unique~\cite[Corollary 4.3b]{schrijver1986theory}. Additionally, if $M$ is an integral matrix then $H$ is a matrix of natural numbers. Let us recall from~\cite[Section~10]{schrijver1986theory} that from a complexity point of view, the matrices $U$ and $H$ are computable in polynomial time and moreover $\det(H)$ divides the determinant of any $r\times r$ non-singular square sub-matrix of $M$. In particular, if the coefficients of $M$ are integers bounded in absolute value by some $m\in\setN$, then $\det(H)$, which is the product of the diagonal elements of $H$ is bounded by $r!m^r$. Let us recall that $(\com(H))_{ax}=\det(H^{i j})$ where $H^{i j}$ is the matrix obtained from $M$ by replacing the $j$th column by a column of zeros, except on the $i$th line in which we put a one. Now, let $\sigma$ be a permutation of $\{1,\ldots,r\}$ and observe that $(H^{i j})_{k\sigma(k)}\leq H_{k k}$ for every $k\in\{1,\ldots,r\}$. It follows that $\prod_{k=1}^r(H^{i j})_{k\sigma(k)}\leq \prod_{k=1}^r H_{k k}=\delta(H)$. Therefore the coefficients of $\com(H)$ are integers bounded in absolute value by $r!\det(H)\leq (r!)^2m^r$.

\medskip

Now, let us prove~\cref{thm:lattice}. We consider a lattice $\vec{L}$ spanned by a sequence $\vec{v}_1,\ldots,\vec{v}_k$ of vectors in $\{-m,\ldots,m\}^d$. We introduce the $d\times k$ matrix $L$ obtained from this sequence by considering $\vec{v}_j$ as the $j$th column of $L$, i.e. $L_{i,j}=\vec{v}_j(i)$ for every $1\leq i\leq d$ and $1\leq j\leq k$. We denote by $r$ the rank of $L$. Recall that $r\leq \min\{d,k\}$. By reordering columns of $L$ (which corresponds to a permutation of $\vec{v}_1,\ldots,\vec{v}_k$) and by reordering the lines of $L$, (which corresponds to a permutation of the components of $\vec{L}$) we can assume without loss of generality that $M$ can be decomposed as follows where $A$ is a $r\times r$ non-singular matrix, $B$ is a $(d-r)\times r$ matrix, $A'$ is a $r\times (k-r)$ matrix, and $B'$ is a $(d-r)\times (k-r)$ matrix.
$$L=\left[\begin{array}{cc}A & A' \\ B & B'\end{array}\right]$$
Let us introduce the $r\times k$ matrix $M\eqdef[AA']$. Notice that $M$ is full row rank. It follows that there exists a uni-modular matrix $U$ such that $MU$ is in Hermite normal form $[H 0]$.

\medskip

We are ready for proving the following lemma.
\begin{lemma}\label{lem:latcar}
  The lattice $\vec{L}$ is the set of vectors $\vec{x}\in\setZ^d$ such that:
  \begin{itemize}
  \item[$(i)$] $\det(H)$ divides the coefficients of $[\vec{x}(1)\ldots\vec{x}(r)]\com(H)$, and
  \item[$(ii)$] $\det(A)[\vec{x}(r+1)\ldots\vec{x}(d)]=[\vec{x}(1)\ldots\vec{x}(r)]\com(A)B^\top$.
  \end{itemize}
\end{lemma}
\begin{proof}
  Since the rank of $L$ is equal to the rank of $A$, it follows that every line of $[BB']$ is a linear combination of the lines of $[A A']$. It follows that there exists a $(d-r)\times r$ rational matrix $C$ such that $[B B']=C[AA']$. Notice that this matrix $C$ satisfies $C=BA^{-1}$ since $A$ is non-singular. 

  \medskip
  
  Assume first that $\vec{x}\in\vec{L}$.

  Since $\vec{x}$ is a linear combination of the column vectors of $L$, and since the $d-r$ last lines of $M$ are linear combinations of the lines of $L$, we deduce that $\vec{x}$ is also a linear combination of the column vectors of $M$. Hence, there exists a sequence $q_1,\ldots,q_k\in\setQ$ such that $\vec{x}=\sum_{j=1}^k q_j\vec{v}_j$. From a matrix point of view, we have  $\left[\begin{array}{c}\vec{x}(1)\\\vdots\\\vec{x}(d)\end{array}\right]=M\left[\begin{array}{c}q_1\\\vdots\\ q_r\end{array}\right]$. Since $M=\left[\begin{array}{c} A\\CA\end{array}\right]$, we deduce that $\vec{x}$ satisfies $(ii)$ by observing that $\det(A)A^{-1}=\com(T)^\top$. 
  
  \medskip

  Next, observe that since $\vec{x}\in\vec{L}$ then $\vec{x}=\sum_{j=1}^k z_j\vec{v}_j$ with $z_1,\ldots,z_k\in\setZ$. It follows that we have:
  \begin{align*}
    \left[\begin{array}{c}\vec{x}(1)\\\vdots\\\vec{x}(r)\end{array}\right]
    =M\left[\begin{array}{c}z_1\\\vdots\\ z_k\end{array}\right]
    =[H 0]U^{-1}\left[\begin{array}{c}z_1\\\vdots\\ z_k\end{array}\right]
    =H\left[\begin{array}{c}z_1'\\\vdots\\ z_r'\end{array}\right]
  \end{align*}
  Where $z_1',\ldots,z_k'$ is the following sequence:
  $$\left[\begin{array}{c}z_1'\\\vdots\\ z_k'\end{array}\right]=U^{-1}\left[\begin{array}{c}z_1\\\vdots\\ z_k\end{array}\right]$$
  We have proved that $\vec{x}$ satisfies $(i)$.

  \medskip

  Finally, assume that $\vec{x}$ is any vector in $\setZ^d$ satisfying $(i)$ and $(ii)$. From $(i)$, we deduce that there exists $z_1',\ldots,z'_r\in\setZ$ such that:
  $$H^{-1}\left[\begin{array}{c}\vec{x}(1)\\\vdots\\\vec{x}(r)\end{array}\right]=\left[\begin{array}{c}z_1'\\\vdots\\ z_r'\end{array}\right]$$
  We introduce the sequence $z_1,\ldots,z_k$ defined as follows where $z'_{r+1},\ldots,z'_k$ are defined as zero:
  $$\left[\begin{array}{c}z_1\\\vdots\\ z_k\end{array}\right]\eqdef U\left[\begin{array}{c}z_1'\\\vdots\\ z_k'\end{array}\right]$$
  Notice that we have:
  $$
  \left[\begin{array}{c}\vec{x}(1)\\\vdots\\\vec{x}(r)\end{array}\right]
  =[H 0]U^{-1}\left[\begin{array}{c}z_1\\\vdots\\ z_k\end{array}\right]\\
  =[AA']\left[\begin{array}{c}z_1\\\vdots\\ z_k\end{array}\right]
  $$
  Combining with $(ii)$ and $[BB']=C[AA']$, we get:
  $$\left[\begin{array}{c}\vec{x}(1)\\\vdots\\\vec{x}(d)\end{array}\right]
  =L\left[\begin{array}{c}z_1\\\vdots\\ z_k\end{array}\right]$$
  Therefore $\vec{x}\in\vec{L}$.
\end{proof}

Now, let $n=\det(H)$. Observe that the coefficients of $\com(H)$ are bounded in absolute value by $(r!)^2m^r$. Since the coefficients of $\com(A)$ are bounded in absolute value by $(r-1)!m^{r-1}$ and the coefficients of $B$ are bounded in absolute value by $m$, we deduce that the coefficients of $\com(A)B^\top$ are bounded in absolute value by $r!m^r$. We have proved~\cref{thm:lattice}.

\section{Words of Petri Net Actions}\label{sec:action}
A \emph{configuration} is a vector in $\setN^d$ and
a \emph{Petri net action} is a pair $a=(\vec{a}_-,\vec{a}_+)$ of configurations.
We introduce $\normi{a}\eqdef \max\{\normi{\vec{a}_-},\normi{\vec{a}_+}\}$.
The \emph{displacement} of $a$ is defined as $\Delta(a)\eqdef \vec{a}_+-\vec{a}_-$.
We associate a Petri net action $a$ with the binary relation $\xrightarrow{a}$
over the configurations defined by
$\vec{x}\xrightarrow{a}\vec{y}$ if for some configuration
$\vec{c}$ we have $(\vec{x},\vec{y})=a+(\vec{c},\vec{c})$.
Given a word $\sigma=a_1\ldots a_k$ of Petri net actions $a_1,\ldots,a_k$, we introduce the \emph{displacement} of $\sigma$ defined as $\Delta(\sigma)\eqdef\sum_{j=1}^k\Delta(a_j)$.
We also denote by $\xrightarrow{\sigma}$ the binary relation over the configurations defined by $\vec{x}\xrightarrow{\sigma}\vec{y}$ if there exists a sequence $\vec{c}_0,\ldots,\vec{c}_k$ of configurations such that $\vec{x}=\vec{c}_0$, $\vec{y}=\vec{c}_k$, and such that $\vec{c}_{j-1}\xrightarrow{a_j}\vec{c}_j$ for every $j\in\{1,\ldots,k\}$.

\medskip

In this section, we provide a way to compile binary relations $\xrightarrow{\sigma}$ for words $\sigma$ of Petri net actions into a quantifier-free Preburger formula. In order to avoid introducing existentially quantified intermediate variables (one for each action of $\sigma$), we recall the definition of \emph{Hurdle} introduced in~\cite[Definition~2.5,p33]{DBLP:phd/ndltd/Hack76}.
First of all, notice that $\vec{x}\xrightarrow{\sigma}\vec{y}$ for some configurations $\vec{x},\vec{y}$ implies $\vec{y}=\vec{x}+\Delta(\sigma)$. It follows that the relation $\xrightarrow{\sigma}$ can be encoded in the Presburger arithmetic as soon as we can characterize with a Presburger formula configurations $\vec{x}$ for which there exists a configuration $\vec{y}$ such that $\vec{x}\xrightarrow{\sigma}\vec{y}$. The \emph{Hurdle} of $\sigma$ provides exactly this characterization. Intuitively, the Hurdle of $\sigma$ is the unique minimal configuration $\vec{x}$ such that $\vec{x}\xrightarrow{\sigma}\vec{y}$ for some $\vec{y}$. More formally, let us introduce the Hurdle function $H$ that maps words of Petri nets actions to configurations defined by induction as follows $H(\varepsilon)$ is the zero configuration, and the following equality for any Petri net action $a=(\vec{a}_-,\vec{a}_+)$ and for any word $\sigma$ of Petri net actions where $\max$ is the component-wise extension of the classical max operator:
$$H(a\sigma)=\max\{\vec{a}_-,H(\sigma)-\Delta(a)\}$$

\medskip

The following lemma shows that the binary relation $\xrightarrow{\sigma}$ can be encoded with a simple quantifier-free Presburger formula. In fact, for any configurations $\vec{x},\vec{y}\in\setN^d$, we have $\vec{x}\xrightarrow{\sigma}\vec{y}$ if, and only if, the following quantifier-free Presburger formula holds:
$$\vec{x}\geq H(\sigma)\wedge \vec{y}=\vec{x}+\Delta(\sigma)$$
\begin{lemma}[\cite{DBLP:phd/ndltd/Hack76}]\label{lem:carastep}
  For every configuration $\vec{x}$ and every word $\sigma$ of Petri net actions, we have $\vec{x}\geq H(\sigma)$ if, and only if, there exists a configuration $\vec{y}$ such that $\vec{x}\xrightarrow{\sigma}\vec{y}$.
\end{lemma}

The following lemma provides bounds on the values occurring in such a formula.
\begin{lemma}
  Let $\sigma$ be a word of Petri net actions in $\{0,\ldots,m\}^d\times\{0,\ldots,m\}^d$ for some natural number $m\in\setN$. Then $\norm{H(\sigma)}_\infty,\norm{\Delta(\sigma)}_\infty\leq |\sigma|m$.
\end{lemma}
\begin{proof}
  By induction on $|\sigma|$.
\end{proof}

\clearpage
\part{Paper Overview}\label{part:overview}
A \emph{Petri net} $A$ (\emph{PN} for short) is a finite
set of Petri net actions.
The \emph{reachability relation} of $A$ is the binary relation $\xrightarrow{A^*}$ over the configurations defined by $\vec{x}\xrightarrow{A^*}\vec{y}$ if there exists a word $\sigma\in A^*$ such that $\vec{x}\xrightarrow{\sigma}\vec{y}$.
The Petri net reachability problem consists in deciding given a PN
$A$ and two configurations $\vec{x}$ and $\vec{y}$ if $\vec{x}\xrightarrow{A^*}\vec{y}$.
Whereas this problem is decidable, its complexity is extremely hard (Ackermannian-complete).
This complexity no longer hold for a natural variant of the reachability problem, called 
the mutual reachability problem, and defined as follows.

\medskip

The mutual reachability relation of a PN $A$ is the binary relation $\xleftrightarrow{A^*}$ defined over the configurations by $\vec{x}\xleftrightarrow{A^*}\vec{y}$ if $\vec{x}\xrightarrow{A^*}\vec{y}$ and $\vec{y}\xrightarrow{A^*}\vec{x}$. Since this relation is an \emph{equivalence relation} (reflexive, symmetric, and transitive), it follows that the set of configurations can be partitioned into equivalence classes. Such an equivalence class is called a \emph{strongly-connected component of configurations} (\emph{SCCC} for short) of $A$. A SCCC $\vec{C}$ is said to be \emph{forward-closed} (resp. \emph{backward-closed}) if for every triple $(\vec{x},a,\vec{y})\in \setN^d\times A\times\setN^d$ such that $\vec{x}\xrightarrow{a}\vec{y}$, we have $\vec{x}\in\vec{C}\Rightarrow \vec{y}\in\vec{C}$ (resp. $\vec{y}\in\vec{C}\Rightarrow\vec{x}\in\vec{C}$). A configuration $\vec{c}_\bot$ is said to be \emph{bottom} (resp. \emph{top}) for a PN $A$ if its SCCC is \emph{forward-closed} (resp. \emph{backward-closed}). 

\medskip

The PN mutual reachability problem consists in deciding given a PN
$A$ and two configurations $\vec{x}$ and $\vec{y}$ if $\vec{x}\xleftrightarrow{A^*}\vec{y}$, equivalently if $\vec{x}$ and $\vec{y}$ are in the same SCCC.
In \cite{concurjournal13}, we proved that the PN mutual reachability
problem is decidable in exponential-space by proving that
there exists at most doubly-exponential long word $u,v\in A^*$ such that $\vec{x}\xrightarrow{u}\vec{y}$ and $\vec{y}\xrightarrow{v}\vec{x}$ when $\vec{x}$ and $\vec{y}$ are mutually reachable.

\medskip

In this paper we focus on the computation of concise quantifier-free Presburger formulas encoding the mutual reachability relations and the set of bottom configurations. Those formulas are obtained by proving that there exist small witnesses of mutual reachability. Those witnesses are defined thanks to the notion of \emph{unfoldings} introduced in \cref{sec:unfoldings}. Intuitively an unfolding of a Petri net is a graph hard-coding the values of some Petri net counters in its states. In this section we also define the subclass of structurally-reversible unfoldings. In \cref{sec:witnesses} we show that this subclass provides witnesses of mutual reachability. From those witnesses we provide in \cref{sec:compiling} a way to compile the mutual reachability relation of a Petri net in the quantifier-free Presburger arithmetic.

\section{Unfoldings}\label{sec:unfoldings}
We introduce in this section the notion of unfoldings defined as \emph{structurally-reversible} graphs hard-coding the values of some Petri net counters in their states.

\medskip

As usual, a (directed) \emph{graph} $G$ is a triple $(Q,A,T)$ where $Q$ is a non empty finite set of states, $A$ is a finite set, and $T$ is a set of \emph{transitions} in $Q\times A\times Q$.
A \emph{path} $\pi$ from a state $p$ to a state $q$
labeled by a word $\sigma\in A^*$ is a
word of transitions in $T^*$ of the form $(q_0,a_1,q_1)\ldots
(q_{k-1},a_k,q_k)$ for some states $q_0,\ldots,q_k$ satisfying
$q_0=p$ and $q_k=q$, and for some actions $a_1,\ldots,a_k$
satisfying $\sigma=a_1\ldots a_k$. A path is said to be
\emph{elementary} if $q_i=q_j$ implies $i=j$. A path such that $q_0=q_k$
is called a \emph{cycle} on $q_0$. A cycle is said to be \emph{simple} if
$q_i=q_j$ with $i<j$ implies $i=0$ and $j=k$. A graph is said to be \emph{strongly-connected} if for every state $p,q\in Q$ there exists a path from $p$ to $q$.

\medskip

A vector (in fact a mapping) in $\setN^I$ where $I$ is a subset of $\{1,\ldots,d\}$ is called an \emph{$I$-configuration}. Given an $I$-configuration $\vec{c}$, we introduce $\normi{\vec{c}}\eqdef\max_{i\in I}|\vec{c}(i)|$.
We associate with a configuration
$\vec{c}\in\setN^d$ the $I$-configuration $\vec{c}|_I$ in $\setN^I$ defined by
$\vec{c}|_I(i)=\vec{c}(i)$ for every $i\in I$.
We also associate with a set $\vec{C}\subseteq \setN^d$ of configurations, the set $\vec{C}|_I\eqdef\{\vec{c}|_I\mid \vec{c}\in\vec{C}\}$.
Given an action
$a=(\vec{a}_-,\vec{a}_+)$ of a Petri net, we extend the binary
relation $\xrightarrow{a}$ over the $I$-configurations by
$\vec{x}\xrightarrow{a}\vec{y}$ if $\vec{x},\vec{y}$ are $I$-configurations
such that there exists an
$I$-configuration $\vec{c}\in\setN^I$ satisfying
$\vec{x}=\vec{a}_-|_I+\vec{c}$ and $\vec{y}=\vec{a}_+|_I+\vec{c}$.

\medskip



An \emph{$I$-unfolding} of a PN $A$ where $I$ is a subset of $\{1,\ldots,d\}$ is a strongly-connected \emph{graph} $G=(Q,A,T)$ where $Q$ is a finite set of $I$-configurations, and $T$ is a set of triples $(p,a,q)\in Q\times A\times Q$ satisfying $p\xrightarrow{a}q$. The \emph{displacement} of a path $\pi$ labeled by a word $\sigma$ is the vector
$\Delta(\pi)\eqdef\Delta(\sigma)$. Given a path $\pi$ from $p$ to $q$ labelled by a word $\sigma$, we denote by $\xrightarrow{\pi}$ the binary relation on the configurations defined by $\vec{x}\xrightarrow{\pi}\vec{y}$ if $\vec{x}|_I=p$,  $\vec{x}\xrightarrow{\sigma}\vec{y}$, and $\vec{y}|_I=q$.

\medskip

An unfolding is said to be \emph{structurally-reversible} if for every transition $t=(p,a,q)$ there exists a path $\pi$ from $q$ to $p$ such that $\Delta(t\pi)=\vec{0}$.
\begin{lemma}\label{lem:stdec}
  Let $G=(Q,A,T)$ be a graph with states $Q\subseteq\setN^I$ encoded in binary. We can decide in polynomial time if $G$ is an $I$-unfolding and we can decide in polynomial time if $G$ is structurally-reversible.
\end{lemma}
\begin{proof}
  Notice that $G$ is an $I$-unfolding if, and only if, $G$ is strongly-connected and $p\xrightarrow{a}q$ for every $(p,a,q)\in T$. This property can be decided in polynomial time. So, we can assume that $G$ is an $I$-unfolding. In that case, notice that $G$ is structurally-reversible, if, and only if, the following linear system over the free variable $f:T\rightarrow\setQ_{>0}$ is satisfiable:
  $$\bigwedge_{q\in Q}\sum_{t\in T\cap \{q\}\times A\times Q}f(t)=\sum_{t\in T\cap Q\times A\times \{q\}}f(t)~~~~~~\wedge ~~~~~~~\sum_{t\in T}f(t)\Delta(t)=\vec{0}$$
  This reduction is a direct application of the Euler's lemma.
\end{proof}

\medskip

We associate a graph $G_{\vec{C},I}\eqdef(Q,A,T)$ with an SCCC $\vec{C}$ and a set $I\subseteq \{1,\ldots,d\}$ such that $\vec{C}|_I$ is finite by $Q\eqdef\vec{C}|_I$ and $T\eqdef\{(\vec{x}|_I,a,\vec{y}|_I) \mid
(\vec{x},a,\vec{y})\in\vec{C}\times A\times\vec{C}\wedge
\vec{x}\xrightarrow{a}\vec{y}\}$.
\begin{lemma}
  The graph $G_{\vec{C},I}$ is a structurally-reversible $I$-unfolding.
\end{lemma}
\begin{proof}
  Let us denote by $G$ the graph $G_{\vec{C},I}$.

  Since $\vec{x}\xrightarrow{a}\vec{y}$ implies $\vec{x}|_I\xrightarrow{a}\vec{y}|_I$, we deduce that for every transition $(p,a,q)\in T$, we have $p\xrightarrow{a}q$.

  Let us show that $G$ is strongly-connected. Let $p,q\in Q$. There exists $\vec{x},\vec{y}\in \vec{C}$ such that $p=\vec{x}|_I$ and $q=\vec{y}|_I$. Since $\vec{C}$ is a SCCC, there exists a word $\sigma$ of actions in $A$ such
  that $\vec{x}\xrightarrow{\sigma}\vec{y}$ and such that all the intermediate configurations are in $\vec{C}$. It follows that there exists a path in $G$ from $p$ to $q$ labeled by $\sigma$. In particular $G$ is strongly-connected.

  Now, let us prove that $G$ is structurally-reversible.
  Let $(p,a,q)$ be a transition in $T$. There exist
  $\vec{x},\vec{y}\in\vec{C}$ such
  that $\vec{x}\xrightarrow{a}\vec{y}$ and such that $p=\vec{x}|_I$
  and $q=\vec{y}|_I$. Moreover since $\vec{C}$
  is a SCCC, there exists a word $\sigma$ of actions in $A$ such
  that $\vec{y}\xrightarrow{\sigma}\vec{x}$ and such that all intermediate configurations are in $\vec{C}$.
  We deduce that there
  exists a path in $G$ from $q$ to $p$ labeled by
  $\sigma$. Notice that
  $\Delta(a)+\Delta(\sigma)=\vec{y}-\vec{x}+\vec{x}-\vec{y}=\vec{0}$. It
  follows that $G$ is structurally-reversible.
\end{proof}

\subsection{Lattice $\vec{L}_G$}
We associate with an $I$-unfolding $G$ the lattice $\vec{L}_G$ spanned by the displacements of the simple cycles of $G$. The following lemma shows that a representation $\gamma_G$ of $\vec{L}_G$ can be computed.
\begin{lemma}\label{lem:repL}
  Let $G=(Q,A,T)$ be an $I$-unfolding with numbers encoded in binary. We can compute in exponential time\footnote{In fact, it can be easily computed in polynomial time thanks to a spanning tree of $G$, but this is out of the scope of that paper.} a representation $\gamma_G$ of $\vec{L}_G$ such that $\norm{\gamma_G}_\infty\leq (d!)^2|Q|^d\norm{A}_\infty^d$.
\end{lemma}
\begin{proof}
  Since a simple cycle has a length bounded by $|Q|$, it follows that displacements of simple cycles are vectors in $\{-m,\ldots,m\}^d$ with $m\leq |Q|.\norm{A}_{\infty}$. We conclude the proof by invoking \cref{thm:lattice}.
\end{proof}

\medskip

The following lemma shows that for every pair $(p,q)$ of states the set $\Delta(\pi)+\vec{L}_G$ does not depend on a path $\pi$ from $p$ to $q$. We denote by $\vec{L}_{p,G,q}$ this set. Notice that $\vec{L}_{p,G,q}$ is a set of the form $\vec{v}+\vec{L}_G$ where $\vec{v}$ is a vector in $\setZ^d$ such that $\normi{\vec{v}}\leq |Q|\norm{A}_\infty^d$ since we can consider for $\vec{v}$ the displacement of an elementary path from $p$ to $q$.
\begin{lemma}
  For every pair $(p,q)$ of states of a strongly-connected unfolding $G$, and for any paths $\alpha,\beta$ from $p$ to $q$, we have $\Delta(\alpha)+\vec{L}_G=\Delta(\beta)+\vec{L}_G$.
\end{lemma}
\begin{proof}
  Since $G$ is strongly-connected, there exists a path $\pi$ from $q$ to $p$. Notice that $\alpha\pi$ and $\beta\pi$ are cycles in $G$, and in particular $\Delta(\alpha\pi)+\vec{L}_G=\vec{L}_G=\Delta(\beta\pi)+\vec{L}_G$. From $\Delta(\alpha)+\Delta(\beta\pi)=\Delta(\beta)+\Delta(\alpha\pi)$ we deduce the lemma.
\end{proof}

\subsection{Upward-closed set $\vec{U}_{q,G}$}
Given a set $\vec{C}$ of configurations, a configuration $\vec{c}\in\vec{C}$ is said to be \emph{minimal} if for every $\vec{x}\in\vec{C}$ we have $\vec{x}\leq \vec{c}\Rightarrow \vec{x}=\vec{c}$. We denote by $\min(\vec{C})$ the set of minimal elements of $\vec{C}$. The \emph{upward-closure} of a set $\vec{B}\subseteq\setN^d$ is the set $\uparrow\vec{B}\eqdef\vec{B}+\setN^d$. Given a configuration $\vec{b}$, we simply denote by $\uparrow\vec{b}$ the set $\uparrow\{\vec{b}\}$. A set $\vec{U}$ of configurations is said to be \emph{upward-closed} if $\uparrow\vec{U}=\vec{U}$. Let us recall that the upward-closure of any set is upward-closed, and since $(\setN^d,\leq)$ is a well quasi ordered set, for any upward-closed set $\vec{U}$, the set $\vec{M}\eqdef\min(\vec{U})$ is finite and satisfies $\vec{U}=\uparrow \vec{M}$. The set $\vec{M}$ is called the \emph{basis} of $\vec{U}$. It follows that the membership of a configuration $\vec{x}$ in $\vec{U}$ is equivalent to the following quantifier-free Presburger formula:
$$\bigvee_{\vec{m}\in\vec{M}}\vec{x}\geq\vec{m}$$

\medskip

A \emph{pumping pair} for $(q,G)$ where $q$ is a state of an $I$-unfolding $G$ is a pair $(u,v)$ of words in $A^*$ that label cycles on $q$ with a length bounded by $db^d$ where $b\eqdef (3dm)^{(d+2)^{2d+1}}$ and $m\eqdef\normi{A}$. We introduce the set $\vec{U}_{q,G}$ of configurations $\vec{c}\in\setN^d$ such that $\vec{c}|_I \geq q$ and there exists a pumping pair $(u,v)$ for $(q,G)$ such that $\vec{c}^-\xrightarrow{u}\vec{c}\xrightarrow{v} \vec{c}^+$ for some configurations $\vec{c}^-,\vec{c}^+$ satisfying $\vec{c}^-(i),\vec{c}^+(i)\geq mr^3(3d rm)^d$ for every $i\not\in I$ where $r\eqdef |Q|$. Notice that $\vec{U}_{q,G}$ is upward-closed. The following lemma provides a way to compute its basis.
\begin{lemma}\label{lem:baseU}
  Let $q$ be a state of an $I$-unfolding $G$ and let $\vec{M}_{q,G}=\min\{\vec{U}_{q,G}\}$. We have $\normi{\vec{M}_{q,G}}\leq s$ where $s\eqdef\max\{\norm{q}_\infty, db^dm+mr^3(3drm)^d\}$. Moreover, we can decide the membership in $\vec{U}_{q,G}$ of a vector $\vec{v}\in\{0,\ldots,s\}^d$ in space $O(\log(|Q|)+d\log(s))$.
\end{lemma}
\begin{proof}
  Let $W$ be the set of words in $A^*$ with a length bounded by $db^d$ that labels a cycle on $q$ in $G$. Let $(u,v)\in W\times W$ and observe that a configuration $\vec{c}\in\setN^d$ is such that $\vec{c}^-\xrightarrow{u}\vec{c}\xrightarrow{v} \vec{c}^+$ for some configurations $\vec{c}^-,\vec{c}^+$ if, and only if, $\vec{c}\geq H(v)$ and $\vec{c}\geq H(u)+\Delta(u)$. Moreover, since in that case $\vec{c}^-=\vec{c}-\Delta(u)$ and $\vec{c}^+=\vec{c}+\Delta(v)$, we deduce that $\vec{c}^-(i),\vec{c}^+(i)\geq mr^3(3d rm)^d$  if, and only if, $\vec{c}(i)\geq  \Delta(u)+mr^3(3d rm)^d$ and $\vec{c}(i)\geq -\Delta(v)+mr^3(3d rm)^d$. Let us introduce the configuration $\vec{c}_{u,v}$ defined by $\vec{c}_{u,v}(i)=q(i)$ if $i\in I$, and by the following equality if $i\not\in I$:
  $$c_{u,v}(i)=
    \max\{ H(v)(i), H(u)(i)+\Delta(u)(i), \Delta(u)(i)+mr^3(3d rm)^d, -\Delta(v)(i)+mr^3(3d rm)^d\}
  $$
  and observe that we have the following equality:
  $$\vec{U}_{q,G}=\uparrow\{\vec{c}_{u,v} \mid (u,v)\in W\times W\}$$
  Denoting by $\vec{M}_{q,G}=\min\{\vec{U}_{q,G}\}$, we deduce that $\norm{\vec{M}_{q,G}}_\infty\leq s$.
  Now, let us consider a vector $\vec{v}\in\{0,\ldots,s\}^d$ and observe that $\vec{v}$ is in $\vec{U}_{q,G}$ if, and only if, there exists $(u,v)\in W\times W$ such that $\vec{v}\geq \vec{c}_{u,v}$. Rather than exploring all the possible pairs $(u,v)\in W\times W$, notice that we can explore step by step pairs of words $(u,v)$, and just compute step by step the values $H(u)$, $H(v)$, $\Delta(u)$, $\Delta(v)$, $|u|$, $|v|$, and the state $q_-,q_+$ such that $u$ is the label of a path from $q_-$ to $q$ and $v$ is the label of a path from $q$ to $q_+$. We then stop if $q_-=q=q_+$ and $\vec{v}\geq \vec{c}_{u,v}$. It follows that we can decide the membership of $\vec{v}$ in $\vec{U}_{q,G}$ in space $O(\log(|Q)+d\log(s))$.
\end{proof}

\section{Witnesses of Mutual Reachability}\label{sec:witnesses}
In this paper, we prove the following characterization of the mutual reachability relation. This characterization will be useful to compute a Presburger formula encoding the mutual reachability relation, and a Presburger formula encoding the set of bottom configurations.
\begin{theorem}\label{thm:witness}
  Let $A$ be a PN, and let $b\eqdef (3dm)^{(d+2)^{2d+1}}$ where $m\eqdef\normi{A}$. A set $\vec{C}$ of configurations are mutually reachable for $A$ if, and only if, there exists a structurally-reversible $I$-unfolding $G=(Q,A,T)$ with $\vec{C}|_I\subseteq Q\subseteq \{q\in\setN^I\mid \normi{q}<b\}$ such that:
  \begin{itemize}
  \item $\vec{c}\in \vec{U}_{\vec{c}|_I,G}$ for every $\vec{c}\in\vec{C}$,
  \item $\vec{y}-\vec{x}\in \vec{L}_{\vec{x}|_I,G,\vec{y}|_I}$ for every $\vec{x},\vec{y}\in\vec{C}$.
  \end{itemize}
\end{theorem}

One way of the previous \cref{thm:witness} is obtained thanks to the following lemma.
\begin{lemma}\label{lem:pathrev}
  Let us consider a structurally-reversible $I$-unfolding $(Q,A,T)$ and let $r\eqdef |Q|$ and $m\eqdef\normi{A}$, let $\vec{x},\vec{y}$ be two configurations such that:
  \begin{itemize}
  \item $p\eqdef \vec{x}|_I$, and  $q\eqdef \vec{y}|_I$ are in $Q$,
  \item $\vec{x}(i),\vec{y}(i)\geq mr^3(3drm)^d$ for every $i\not\in I$, and
  \item $\vec{y}-\vec{x}\in\vec{L}_{p,G,q}$.
  \end{itemize}
  Then for any elementary path $\pi$ from $p$ to $q$, there exists a cycle $\theta$ on $q$
  such that $\vec{x}\xrightarrow{\pi\theta}\vec{y}$ and
  such that
  $|\theta|\leq \norm{\vec{y}-\vec{x}-\Delta(\pi)}2r^3(3drm)^{2d}$.
\end{lemma}
\begin{proof}
  The proof is given in \cref{part:pathrev}.
\end{proof}

\medskip

In fact, let us consider a set $\vec{C}$ of configurations such that there exists an $I$-unfolding $G=(Q,A,T)$ such that $\vec{C}|_I\subseteq Q\subseteq \{q\in\setN^I\mid \normi{q}<b\}$ such that:
\begin{itemize}
\item $\vec{c}\in \vec{U}_{\vec{c}|_I,G}$ for every $\vec{c}\in\vec{C}$,
\item $\vec{y}-\vec{x}\in \vec{L}_{\vec{x}|_I,G,\vec{y}|_I}$ for every $\vec{x},\vec{y}\in\vec{C}$.
\end{itemize}
  
  We introduce $r\eqdef|Q|$. Let $\vec{x},\vec{y}\in \vec{C}$. Let $p\eqdef \vec{x}|_I$, $q=\vec{y}|_I$. Since $\vec{x}\in\vec{U}_{p,G}$, there exists a cycle $\alpha$ on $p$ with a length bounded by $db^d$ such that $\vec{x}\xrightarrow{\alpha}\vec{x}^+$ for some configurations $\vec{x}^+$ such that $\vec{x}^+(i)\geq  mr^3 (3drm)^d$ for every $i\not\in I$. Symmetrically, there exists a cycle $\beta$ on $q$ with a length bounded by $db^d$ such that $\vec{y}^-\xrightarrow{\beta}\vec{y}$ for some configurations $\vec{y}^-$ such that $\vec{y}^-(i)\geq  mr^3 (3drm)^d$ for every $i\not\in I$. Observe that $\vec{y}^--\vec{x}^+=\vec{y}-\vec{x}-\Delta(\beta)+\Delta(\alpha)$. In particular $\vec{y}^--\vec{x}^+-\Delta(\pi)$ is in the lattice $\vec{L}_G$. \cref{lem:pathrev} shows that there exists a cycle $\theta$ on $q$ such that $\vec{x}^+\xrightarrow{\pi\theta}\vec{y}^-$ and such that
  $|\theta|\leq \norm{\vec{y}^--\vec{x}^+-\Delta(\pi)}2r^3(3drm)^{2d}$. 
  It follows that we have $\vec{x}\xrightarrow{\pi'}\vec{y}$ with $\pi'\eqdef \alpha\pi\theta\beta$. By symmetry, we get $\vec{x}\xleftrightarrow{A^*}\vec{y}$.

\begin{remark}
  As a direct consequence of the previous proof, notice that we can provide a bound on the length of a path from $\vec{x}$ to $\vec{y}$ that only depends on $\norm{\vec{y}-\vec{x}}$, $d$ and $m$, a result proved~\cite{DBLP:conf/fsttcs/Leroux19}. In fact, notice that $\vec{y}^--\vec{x}^+-\Delta(\pi)=\vec{y}-\vec{x}-\Delta(\beta)+\Delta(\alpha)-\Delta(\pi)$. It follows that $\norm{\vec{y}^--\vec{x}^+-\Delta(\pi)}\leq\norm{\vec{y}-\vec{x}}+dm(|\pi|+|\beta|+|\alpha|)\leq \norm{\vec{y}-\vec{x}}+dm(2d+1)b^d)$ since $|\pi|\leq r\leq b^d$, $|\beta|,|\alpha|\leq db^d$. It follows that:
  $$|\pi'|\leq (2d+1)b^d+\norm{\vec{y}-\vec{x}}+dm(2d+1)b^d)2r^3(3drm)^{2d}$$
  By observing that $r\leq b^d$ and $b=(3dm)^{(d+2)^{2d+1}}$, we deduce that there exists a constant $c_{d,m}$ that only depends on $d$ and $m$ such that $|\pi'|\leq \norm{\vec{y}-\vec{x}}c_{d,m}$.
\end{remark}

\medskip

The other way of the previous \cref{thm:witness} is obtained thanks to the following lemma.
\begin{lemma}\label{thm:witnessex}
  Let $A$ be a PN, and let $b\eqdef (3dm)^{(d+2)^{2d+1}}$ where $m\eqdef\normi{A}$. For every SCCC $\vec{C}$ of $A$, there exists a set $I\subseteq\{1,\ldots,d\}$ such that $\vec{C}|_I\subseteq \{q\in\setN^I\mid \normi{q}<b\}$, and denoting by $G$ the structurally-reversible $I$-unfolding $G_{\vec{C},I}$:
  \begin{itemize}
  \item We have $\vec{c}\in \vec{U}_{\vec{c}|_I,G}$ for every $\vec{c}\in\vec{C}$, and
  \item we have $\vec{y}-\vec{x}\in \vec{L}_{\vec{x}|_I,G,\vec{y}|_I}$ for every $\vec{x},\vec{y}\in \vec{C}$.
  \end{itemize}
\end{lemma}
\begin{proof}
  The proof is given in \cref{part:witnessex}.
\end{proof}

Concerning the set of bottom configurations, we provide the following theorem. An $I$-unfolding $G=(Q,A,T)$ is said to be \emph{forward-closed } if for every $p\in Q$ and for every $a\in A$, if there exists an $I$-configuration $q$ such that $p\xrightarrow{a}q$, then $q\in Q$ and $(p,a,q)\in T$. 
\begin{theorem}\label{thm:bot}
  Let $A$ be a PN, and let $b\eqdef (3dm)^{(d+2)^{2d+1}}$ where $m\eqdef\normi{A}$. A configuration $\vec{c}$ is bottom if, and only if, there exists a forward-closed structurally-reversible $I$-unfolding $G=(Q,A,T)$ such that $\normi{q}<b$ for every $q\in Q$, a state $r\in Q$ such that $\vec{c}|_I=r$, $\vec{c}\in \vec{U}_{r,G}$, and such that for every $(p,a,q)\in T$ and for every $\vec{v}\in\vec{L}_{r,G,p}$, we have:
  $$\vec{c}+\vec{v}\in \vec{U}_{p,G}\cap \uparrow \vec{a}_-~~\Longrightarrow~~\vec{c}+\vec{v}+\Delta(a)\in \vec{U}_{q,G}$$  
\end{theorem}
\begin{proof}
  Assume first that $\vec{c}$ is a bottom configuration and let $\vec{C}$ be its SCCC. \cref{thm:witnessex} shows that there exists a set $I\subseteq \{1,\ldots,d\}$ such that $\normi{\vec{x}|_I}<b$ for every $\vec{x}\in \vec{C}$, and denoting by $G$ the structurally-reversible $I$-unfolding $G_{\vec{C},I}$:
  \begin{itemize}
  \item We have $\vec{x}\in \vec{U}_{\vec{x}|_I,G}$ for every $\vec{x}\in\vec{C}$.
  \item We have $\vec{y}-\vec{x}\in \vec{L}_{\vec{x}|_I,G,\vec{y}|_I}$ for every $\vec{x},\vec{y}\in \vec{C}$.
  \end{itemize}
  Let $r=\vec{c}|_I$.
  
  Let us prove that $G$ is forward-closed. Let $p\in Q$, $a\in A$, and consider an $I$-configuration $q$ such that $p\xrightarrow{a}q$ and let us prove that $q\in Q$. There exists $\vec{x}\in\vec{C}$ such that $p=\vec{x}|_I$. Since $\vec{x}\in \vec{U}_{p,G}$, there exists a configuration $\vec{x}^+$ reachable from $\vec{x}$ such that $\vec{x}^+|_I=r$ and $\vec{x}^+(i)\geq m$ for every $i\not\in I$. Since $\vec{C}$ is bottom, it follows that $\vec{x}^+\in \vec{C}$. So, by replacing $\vec{x}$ by $\vec{x}^+$ we can assume that $\vec{x}=\vec{x}^+$. As $p\xrightarrow{a}q$ and $\vec{x}|_I=p$, we get $\vec{x}(i)=p(i)\geq \vec{a}_-(i)$ for every $i\in I$. Moreover, as $\vec{x}(i)\geq m\geq \vec{a}_-(i)$ for every $i\not\in I$, we have proved that $\vec{x}\geq\vec{a}_-$. Hence, $\vec{x}\xrightarrow{a}\vec{y}$ with $\vec{y}\eqdef\vec{x}+\Delta(a)$. As $\vec{C}$ is a bottom SCCC, we deduce that $\vec{y}\in\vec{C}$. Hence $(\vec{x}|_I,a,\vec{y}|_I)\in T$. Since this triple is $(p,a,q)$, we have proved that $q\in Q$ and $(p,a,q)\in T$. Hence $G$ is forward-closed.

  Let us consider $(p,a,q)\in T$ and $\vec{v}\in\vec{L}_{r,G,p}$ such that $\vec{c}+\vec{v}\in \vec{U}_{p,G}\cap \uparrow \vec{a}_-$. Let $\vec{x}\eqdef\vec{c}+\vec{v}$. Theorem~\ref{thm:witness} shows that $\vec{x}\in \vec{C}$. Since $\vec{x}\geq\vec{a}_-$ we deduce that $\vec{x}\xrightarrow{a}\vec{y}$ with $\vec{y}\eqdef\vec{x}+\Delta(a)$. Since $\vec{C}$ is a bottom SCCC, it follows that $\vec{y}\in\vec{C}$. Hence $\vec{y}\in \vec{U}_{q,G}$ since $\vec{y}|_I=q$. We have proved one direction of the theorem.

  Now, assume that $\vec{c}$ is a configuration such that there exists a forward-closed structurally-reversible $I$-unfolding $G=(Q,A,T)$, a state $r\in Q$ such that $\vec{c}|_I=r$, $\vec{c}\in \vec{U}_{r,G}$, and such that for every $(p,a,q)\in T$ and for every $\vec{v}\in\vec{L}_{r,G,p}$ we have:
  $$\vec{c}+\vec{v}\in \vec{U}_{p,G}\cap \uparrow \vec{a}_-~~\Longrightarrow~~\vec{c}+\vec{v}+\Delta(a)\in \vec{U}_{q,G}$$  
  And let us prove that $\vec{c}$ is bottom. It is sufficient to prove that for every configuration $\vec{x}$ reachable from $\vec{c}$, the configurations $\vec{x}$ and $\vec{c}$ are mutually reachable. There exists a word $\sigma\in A^*$ such that $\vec{c}\xrightarrow{\sigma}\vec{x}$. Assume that $\sigma=a_1\ldots a_k$, and let us introduce the sequence $\vec{c}_0,\ldots,\vec{c}_k$ of configurations such that $\vec{c}_0=\vec{x}$, $\vec{c}_k=\vec{x}$ and such that $\vec{c}_{i-1}\xrightarrow{a_i}\vec{c}_i$ for every $1\leq i\leq k$. Let us introduce $q_i=\vec{c}_i|_I$. Since $G$ is forward-closed, and $q_{i-1}\xrightarrow{a_i}q_i$ for every $1\leq i\leq k$, we deduce that $q_i\in Q$ for every $0\leq i\leq k$. Notice that $\vec{c}_0\in \vec{U}_{q_0,G}$. Assume by induction that $\vec{c}_{i-1}\in \vec{U}_{q_{i-1},G}$ for some $i\in\{1,\ldots,k\}$ and let us prove that $\vec{c}_i\in \vec{U}_{q_i,G}$. Observe that $\vec{c}_{i-1}=\vec{c}+\Delta(a_1\ldots a_{i-1})$. As $a_1\ldots a_{i-1}$ is the label of a path from $r$ to $q_{i-1}$, it follows that $\vec{v}\eqdef  \Delta(a_1\ldots a_{i-1})$ is in $\vec{L}_{r,G,p}$. Since additionally we have $\vec{c}_{i-1}\geq (\vec{a}_i)_-$ we deduce that $\vec{c}+\vec{v}+\Delta(a_i)\in \vec{U}_{q_i,G}$. As $\vec{c}+\vec{v}+\Delta(a_i)=\vec{c}_i$, we have proved the induction.
  In particular $\vec{x}=\vec{c}_k$ is in $\vec{U}_{q_k,G}$. From~\cref{thm:witness} we deduce that $\vec{x}$ and $\vec{c}$ are mutually reachable. Hence $\vec{c}$ is a bottom configuration.
\end{proof}

\section{Compiling in Presburger}\label{sec:compiling}
By encoding with a quantifier-free Presburger formula the membership of a vector in the upward-closed set $\vec{U}_{q,G}$ and the lattice $\vec{L}_G$, we obtain as a direct corollary the following theorem.
\begin{theorem}\label{thm:main}
  Let $A$ be a PN, and let $s\eqdef 2mb^{3d}(3db^dm)^d$ where $b\eqdef (3dm)^{(d+2)^{2d+1}}$ and $m\eqdef\normi{A}$. There exists a set $S_A$ of tuples $(\vec{a},\vec{b},\vec{v},\gamma)$ where $\vec{a},\vec{b}\in\{0,\ldots,s\}^d$ and $\vec{v}\in\{-s,\ldots,s\}^d$ and $\gamma$ is a representation of a lattice such that $\normi{\gamma}\leq s$ with a membership problem in space $O(s)$ such that for every configuration $\vec{x},\vec{y}\in\setN^d$, we have:
  $$\vec{x}\xleftrightarrow{A^*}\vec{y}~~~\Longleftrightarrow~~~\bigvee_{(\vec{a},\vec{b},\vec{v},\gamma)\in S_A}\vec{x}\geq \vec{a}\wedge \vec{y}-\vec{x}-\vec{v}\in \crochet{\gamma}\wedge \vec{y}\geq \vec{b}$$
\end{theorem}
\begin{proof}
  Let us introduce the set $S_A$ of tuples $(\vec{a},\vec{b},\vec{v},\gamma)$ such that there exists a structurally-reversible $I$-unfolding $G$ satisfying $Q\subseteq \{q\in\setN^I\mid \normi{q}<b\}$, let $r=|Q|\leq b^d$, such that $\gamma$ is a representation of $\vec{L}_G$ satisfying $\normi{\gamma}\leq (d!)^2b^dm^d\leq s$ computed by some given algorithm (see \cref{lem:repL}), two states $p,q\in Q$, a simple path $\pi$ from $p$ to $q$ satisfying $\vec{v}=\Delta(\pi)$, and $\vec{a},\vec{b}\in \{0,\ldots,s\}^d$ satisfying $\vec{a}\in\vec{U}_{p,G}$ and $\vec{b}\in\vec{U}_{q,G}$. Notice that $\normi{\vec{v}}\leq b^dm\leq s$. From \cref{thm:witness}, we deduce that $S_A$ satisfies the theorem.

  Now, just observe we can decide if a tuple $(\vec{a},\vec{b},\vec{v},\gamma)$ where $\vec{a},\vec{b},\vec{v}\in\{0,\ldots,s\}^d$ and $\gamma$ is a representation of a lattice such that $\normi{\gamma}\leq s$ is in $S_A$ by enumerating all the possible structurally-reversible $I$-unfolding $G$ satisfying $Q\subseteq \{q\in\setN^I\mid \normi{q}<b\}$ (we just remember one at each step of the enumeration). Then, we just compute with the algorithm used for defining $\vec{S}_A$ a representation of $\vec{L}_G$ and check if $\gamma$ is this representation. Then we can check if $\vec{a}\in\vec{U}_{p,G}$ and $\vec{b}\in\vec{U}_{q,G}$ in space $O(\log(|Q|)+d\log(s))$. We are done.
\end{proof}

From a Presburger formula $\phi_A$ encoding the mutual reachability relation, a Presburger formula $\phi_A^\bot(\vec{c})$ encoding the set of bottom configurations can be obtained as follows:
$$\phi_A^\top(\vec{c})~\eqdef~\forall \vec{x}\bigwedge_{a\in A}\phi_A(\vec{c},\vec{x})\wedge \vec{x}\geq \vec{a}_-\Rightarrow \phi_A(\vec{c},\vec{x}+\Delta(a))$$
Even if such a formula is rather simple, it does not take advantage of the fact that $\vec{c}$, $\vec{x}$ and $\vec{x}+\Delta(a)$ are in the same SCCC, a property used in the following theorem.

\medskip

A $k$-\emph{threshold} formula $\phi$ is a boolean combination of formulas of the form $\vec{x}(i)\geq z$ where $i\in\{1,\ldots,d\}$ and $z\in\setZ$ is an integer satisfying $|z|\leq k$. The size of such a formula is the one expected with numbers encoded in binary.
\begin{theorem}\label{thm:bot}
  Let $A$ be a PN and let $s\eqdef 2mb^{3d}(3db^dm)^d$ where $b\eqdef (3dm)^{(d+2)^{2d+1}}$ and $m\eqdef\normi{A}$. We can compute in time $O(s^d)$ a set $T_A$ of tuples $(r,\gamma,\phi)$ where $r$ is an $I$-configuration with $\normi{q}<b$, $\gamma$ is a representation of a lattice such that $\normi{\gamma}\leq s$, and $\phi$ is a $k$-threshold formula with a size bounded by $O(s)$ and $k\leq s$, and such that a configuration $\vec{c}$ is a bottom configurations if, and only if:
  $$\bigvee_{(q,\gamma,\phi)\in T_A}\vec{c}|_I=r~\wedge~\forall \vec{v}\in\crochet{\gamma}~ \phi(\vec{c}+\vec{v})$$
\end{theorem}
\begin{proof}
  The set $T_A$ is obtained by enumerating the forward-closed structurally-reversible $I$-unfoldings $G=(Q,A,T)$ such that $\normi{q}<b$ for every $q\in Q$. We introduce $\vec{M}_{q,G}\eqdef\min(\vec{U}_{q,G})$ for every $q\in Q$. For such a $G$ and for each state $r\in Q$, we introduce a sequence $(v_p)_{p\in Q}$ of vectors such that $\vec{v}_p$ is the label of an elementary path from $r$ to $p$. We denote by $\phi_{r,G}$ the following threshold formula:

  $$
  \phi_{r,G}(\vec{x})
  \eqdef
  \bigwedge_{(p,a,q)\in T}((\bigvee_{\vec{m}\in\vec{M}_{p,G}}\vec{x}\geq \max(\vec{m},\vec{a}_-)-\vec{v}_p)\Rightarrow (\bigvee_{\vec{m}\in\vec{M}_{q,G}}\vec{x}\geq \vec{m}-\Delta(a)-\vec{v}_p)) 
  $$
  From \cref{thm:bot} we deduce that $T_A$ satisfies the theorem.
\end{proof}

\clearpage
\part{Proof of \cref{lem:pathrev}}\label{part:pathrev}\label{part:lem1}

In this part, we prove \cref{lem:pathrev}. All other results proved in this section are not used in the sequel. The proof follows an extended form of the \emph{zigzag-freeness}
approach introduced in~\cite{LerouxS04}. Intuitively, we prove that the cycle $\theta$ can be obtained by concatenating a
sequence $\theta_1,\ldots,\theta_k$ of short cycles on $q$ such that
for every $n\in\{0,\ldots,k\}$ the displacement of $\Delta(\theta_1\ldots \theta_n)$ is almost the
vector $\frac{n-d}{k}(\vec{y}-\vec{x}-\Delta(\pi))$.

\medskip

\section{Reordering finite sums of integer vectors}
In this section, we show that if a vector $\vec{z}\in\setZ^d$ is the sum of a sequence $\vec{z}_1,\ldots,\vec{z}_k\in\setZ^d$, then we can extract a sub-sequence satisfying the same property and such that additionally $k$ is small (with respect to some parameters). Moreover, we also prove that we can reorder such a sequence in such a way $\sum_{j=1}^n\vec{z}_j\geq \min\{\vec{z}(i),0\}$ for every $i\in\{1,\ldots,d\}$ and for every $n\in\{0,\ldots,k\}$.

\medskip

Those two results are obtained thanks to following central result.
\begin{lemma}[\cite{Grinberg1980}]\label{lem:steinitz}
  Let $\vec{v}_1,\ldots,\vec{v}_k$ be a non-empty sequence of vectors in
  $\mathbb{R}^d$ such that $\normi{\vec{v}_j}\leq 1$ for every $1\leq
  j\leq k$ and let
  $\vec{v}=\sum_{j=1}^k\vec{v}_j$. There exists a permutation
  $\sigma$ of $\{1,\ldots,k\}$ such that for every $n\in\{d,\ldots, k\}$, we
  have:
  $$\normi{\sum_{j=1}^n\vec{v}_{\sigma(j)}- \frac{n-d}{k}\vec{v}}\leq d$$
\end{lemma}

In fact, from the previous lemma we deduce the following two corollaries.
\begin{corollary}\label{cor:stei1}
  Assume that $\vec{z}=\vec{z}_1+\ldots+\vec{z}_k$ for some vectors $\vec{z}_1,\ldots,\vec{z}_k\in\setZ^d$. Then there exists $J\subseteq \{1,\ldots,k\}$ such that $\vec{z}=\sum_{j\in J}\vec{z}_j$ with $|J|\leq 2\norm{\vec{z}}(3dm)^d$ and $m\eqdef\max_{1\leq j\leq k}\normi{\vec{z}_j}$.
\end{corollary}
\begin{proof}
  Assume that $\vec{z}=\vec{z}_1+\ldots+\vec{z}_k$ for some vectors $\vec{z}_1,\ldots,\vec{z}_k\in\setZ^d$ and assume that there does not exists a set $J$ strictly smaller than $\{1,\ldots,k\}$ such that $\vec{z}=\sum_{j\in J}\vec{z}_j$. This last property is equivalent to $\sum_{j\in J}\vec{z}_j\not=\vec{0}$ for every non-empty subset $J\subseteq\{1,\ldots,k\}$. We introduce $m\eqdef\max_{1\leq j\leq k}\normi{\vec{z}_j}$. Let us prove that $k\leq 2\norm{\vec{z}}(3dm)^d$. Without loss of generality, we can assume that $\vec{z}\geq 0$ since we can swap the sign of $\vec{z}(i),\vec{z}_1(i),\ldots,\vec{z}_k(i)$ for any $i$ to reduce our problem to this special case. If $k=0$ the lemma is proved, so let us assume that $k\geq 1$. In particular $\norm{\vec{z}}\geq 1$ and $m\geq 1$.
  
  Let us first prove that there exists a sequence $\vec{e}_1,\ldots\vec{e}_k$ of configurations such that $\vec{e}_j\leq \max\{\vec{0},\vec{z}_j\}$ for every $j\in\{1,\ldots,k\}$, and such that $\vec{z}=\sum_{j=1}^k\vec{e}_j$. To do so, we introduce the non-decreasing sequence $\vec{c}_0,\ldots,\vec{c}_k$ of configurations defined as $\vec{c}_0\eqdef\vec{0}$ and by induction for every $j\in\{1,\ldots,k\}$ by $\vec{c}_j\eqdef\max\{\vec{c}_{j-1},\vec{c}_{j-1}+\vec{z}_j\}$. By induction, we observe that $\vec{c}_j\geq \sum_{\ell=1}^j\vec{z}_j$ for every $j\in\{0,\ldots,k\}$. In particular $\vec{c}_k\geq \vec{z}$. We also introduce the sequence $\vec{e}_1,\ldots,\vec{e}_k$ of configurations defined by $\vec{e}_j\eqdef\min\{\vec{z},\vec{c}_j\}-\min\{\vec{z},\vec{c}_{j-1}\}$. Notice $\sum_{j=1}^k\vec{e}_j=\min\{\vec{z},\vec{c}_k\}-\min\{\vec{z},\vec{c}_0\}$. Since $\vec{c}_k\geq \vec{z}$ and $\vec{c}_0=\vec{0}$, we derive $\sum_{j=1}^k\vec{e}_j=\vec{z}$. Now, let us prove that $\vec{e}_j\leq \max\{\vec{0},\vec{z}_j\}$. So, let $i\in\{1,\ldots,d\}$. Assume first that $\vec{z}_j(i)\leq 0$. In that case from $\vec{c}_j\eqdef\max\{\vec{c}_{j-1},\vec{c}_{j-1}+\vec{z}_j\}$, we deduce that $\vec{c}_j(i)=\vec{c}_{j-1}(i)$. It follows that $\vec{e}_j(i)=0$ and we are done. Now assume that $\vec{z}_j(i)>0$. In that case from $\vec{c}_j=\max\{\vec{c}_{j-1},\vec{c}_{j-1}+\vec{z}_j\}$ we deduce that $\vec{c}_j(i)=\vec{c}_{j-1}(i)+\vec{z}_j(i)$. From $\vec{e}_j\eqdef\max\{\vec{z},\vec{c}_j\}-\max\{\vec{z},\vec{c}_{j-1}\}$ we deduce that $\vec{e}_j(i)=\max\{\vec{z}(i),\vec{c}_{j-1}(i)+\vec{z}_j(i)\}-\max\{\vec{z}(i),\vec{c}_{j-1}(i)\}\leq \vec{z}_j(i)$ and we are done.
  
  Next, let us introduce the sequence
  $\vec{v}_1,\ldots,\vec{v}_k$ defined by
  $\vec{v}_j\eqdef \vec{z}_j-\vec{e}_j$. Notice that
  $\normi{\vec{v}_j}\leq m$ and
  $\sum_{j=1}^k\vec{v}_j=\vec{0}$. We introduce
  $\vec{x}_n\eqdef\sum_{j=1}^n\vec{v}_j$. 
  By applying a
  permutation, \cref{lem:steinitz} applied on the sequence
  $(\frac{1}{m}\vec{v}_j)_{1\leq j\leq n}$ shows that we can assume without
  loss of generality that $\vec{x}_n\in\vec{X}$ for every $d\leq n\leq
  k$ where $\vec{X}$ is the
  set of vectors $\vec{x}\in\setZ^d$ such that
  $\normi{\vec{x}}\leq md$. Notice that if $n\in\{0,\ldots,d\}$, we
  also have $\vec{x}_n\in \vec{X}$ since $\vec{x}_n$ is a sum of at
  most $d$ vectors with a norm bounded by $m$.
  
  The cardinal of
  $\vec{X}$ is bounded by $(1+2dm)^d\leq (3dm)^d$. Now, assume by
  contradiction that there
  exists $\ell\in \{0,\ldots,k-(3dm)^d\}$ satisfying
  $\vec{e}_j=\vec{0}$ for every $j\in
  \{\ell+1,\ldots,\ell+(3dm)^d\}$. Notice that there exists $p<q$ in
  $\{\ell,\ldots,\ell+(3dm)^d\}$ such that $\vec{x}_p=\vec{x}_q$
  since the cardinal of $\vec{X}$ is bounded by $(3dm)^d$. It
  follows that $\sum_{j=p+1}^q\vec{v}_j=\vec{0}$. From $\vec{e}_j=\vec{0}$ for every $j\in
  \{\ell+1,\ldots,\ell+(3dm)^d\}$ it follows that
  $\vec{v}_j=\vec{z}_j$ for every $j\in\{p+1,\ldots,q\}$. In
  particular $\sum_{j=p+1}^q\vec{z}_j=\vec{0}$.
  Hence $k$ is
  not minimal since we can remove the vectors
  $\vec{z}_{p+1},\ldots,\vec{z}_q$ from the sequence
  $\vec{z}_1,\ldots,\vec{z}_k$, and we get a contradiction.
  It follows that for every
  $\ell\in \{0,\ldots,k-(3dm)^d\}$ there exists $j\in
  \{\ell+1,\ldots,\ell+(3dm)^d\}$ such that
  $\vec{e}_j\not=\vec{0}$. From
  $\norm{\vec{z}}=\sum_{j=1}^k\norm{\vec{e}_j}$, it
  follows that $\norm{\vec{z}}\geq \frac{k}{(3dm)^d}-1$. Hence $k\leq
  (\norm{\vec{z}}+1)(3dm)^d$. Since $1+\norm{\vec{z}}\leq 2\norm{\vec{z}}$, we deduce that $k\leq 2\norm{\vec{z}}(3dm)^d$.
\end{proof}

\begin{corollary}\label{cor:stei2}
  Assume that $\vec{z}=\vec{z}_1+\cdots+\vec{z}_k$ where
  $\vec{z}_1,\ldots,\vec{z}_k\in\setZ^d$. There exists a permutation
  $\sigma$ of $\{1,\ldots,k\}$ such that for every $n\in\{0,\ldots,
  k\}$ and for every $i\in\{1,\ldots,d\}$, we
  have:
  $$\sum_{j=1}^n\vec{z}_{\sigma(j)}(i)\geq \min\{\vec{z}(i),0\}-md$$
  where $m\eqdef\max_j\normi{\vec{z}_j}$.
\end{corollary}
\begin{proof}
  If $k=0$ the lemma is
  proved. So, we can assume that $k\geq 1$, and in particular $m\geq 1$.
  By applying a
  permutation, \cref{lem:steinitz} on the sequence $(\frac{1}{m}\vec{z}_j)_{1\leq
    j\leq k}$ shows that we can assume without
  loss of generality that for every $n\in\{0,\ldots,k\}$, there exists
  a vector $\vec{e}_n\in\mathbb{R}^d$ such that
  $\normi{\vec{e}_n}\leq md$ and such that
  $\vec{x}_n=\frac{n-d}{k}\vec{z}+\vec{e}_n$ where $\vec{x}_n\eqdef
  \sum_{j=1}^n\vec{z}_j$. Let $i\in\{1,\ldots,d\}$ and let us prove
  that $\vec{x}_n(i)\geq \min\{\vec{z}(i),0\}-md$. Observe that if
  $n\in\{0,\ldots,d\}$ then the property is immediate since
  $\vec{x}_n(i)\geq -md$. So, let us assume that $n>d$. If
  $\vec{z}(i)\geq 0$ then $\frac{n-d}{k}\vec{z}(i)\geq 0$ and we get
  $\vec{x}_n(i)\geq \vec{e}_n(i)\geq -md$. If $\vec{z}(i)\leq 0$ then
  $\frac{n-d}{k}\vec{z}(i)\geq \vec{z}(i)$. In particular
  $\vec{x}_n(i)\geq \min\{\vec{z}(i),0\}-md$ also in that case.
\end{proof}

\section{From simple cycles to small full-state cycles}
A cycle of an unfolding $G$ is said to be \emph{full-state} if every state of $G$ occurs in the cycle. In this section we prove that if $G$ is structurally-reversible, then the displacement of any simple cycle is the displacement of a ``small'' full-state cycle. In this section $G$ is a structurally-reversible unfolding.

\medskip

We first observe that the negation of the displacement of any cycle is the
displacement of another cycle as shown by the following lemma.
\begin{lemma}\label{lem:reverse}
  For every cycle $\theta$, there exists a cycle $\theta'$ such that
  $\Delta(\theta')=-\Delta(\theta)$.
\end{lemma}
\begin{proof}
  Assume that $\theta=t_1\ldots t_k$ for some transitions $t_1,\ldots,t_k$.
  Since $G$
  is structurally-reversible, for every $j\in\{1,\ldots,k\}$, there
  exists a path $\pi_j$ such that $t_j\pi_j$ is a cycle with a zero displacement.
  Now, observe that
  $\theta'\eqdef \pi_k\ldots\pi_1$ is a cycle such that
  $\Delta(\theta')=-\Delta(\theta)$.
\end{proof}

Let us show the following lemma based on small solutions for linear integer programming~\cite{P-RTA91}.
\begin{lemma}\label{lem:pottier0}
  Every transition occurs in a finite sequence
  $\theta_1,\ldots,\theta_n$ of simple cycles such that
  $\Delta(\theta_1)+\cdots+\Delta(\theta_n)=\vec{0}$ and such that $n\leq (3drm)^{d}$
\end{lemma}
\begin{proof}
  Let $t$ be a transition. Since $G$ is strongly connected, the
  transition $t$ occurs in a simple cycle
  $\theta_0$. \cref{lem:reverse} shows that $-\Delta(\theta_0)$ is a
  finite sum of displacements of simple cycles. In particular
  $-\Delta(\theta_0)$ is in the cone generated by the displacements of
  simple cycles, i.e. the finite sums of displacements of
  simple cycles multiplied by non-negative rational numbers.
  From Carathéodory theorem, there exists $d$ simple
  cycles $\theta_1,\ldots,\theta_d$ and $d$ non-negative
  rational numbers $r_1,\ldots,r_d$ such that
  $-\Delta(\theta_0)=\sum_{j=1}^dr_j\Delta(\theta_j)$. By introducing
  a positive integer $h_0$ such that $h_j\eqdef h_0 r_j$ is a natural
  number for every $j$, we derive that the
  following linear system over the sequences $(h_j)_{0\leq\ j\leq
    d}$ of natural numbers 
  $$\sum_{j=0}^dh_j\vec{v}_j=\vec{0}$$
  admits a solution satisfying
  $h_0>0$ where $\vec{v}_j\eqdef\Delta(\theta_j)$.

  \medskip
  
  From~\cite{P-RTA91}, it follows that solutions of that system can be
  decomposed as finite sums of ``minimal'' solutions 
  $(h_j)_{1\leq j\leq k}$ of the same system satisfying
  additionally the following constraint:
  $$\sum_{j=0}^dh_j\leq (1+(d+1)rm)^d$$
  From $1+(d+1)rm\leq (3drm)$, we derive
  $(1+(d+1)rm)^d\leq (3drm)^{d}$.
  Since there exist solutions of that system with $h_0>0$, there
  exists at least a minimal one satisfying the same constraint. We
  have proved the lemma.
\end{proof}

We deduce the following lemma.
\begin{lemma}\label{lem:t1}
  There exists a full-state cycle with a zero displacement with a length
  bounded by $r^2(r-1)(3drm)^{d}$.
\end{lemma}
\begin{proof}
  Let us consider the set $H$ of pairs $(p,q)\in Q\times Q$ such that
  there exists a transition from $p$ to $q$ with $p\not=q$. For every
  $h\in H$ of the form $(p,q)$, we select a transition $t_h\in T$ from
  $p$ to $q$. \cref{lem:pottier0} shows that for every $h\in H$,
  there exists a sequence of at most $(3drm)^{d}$ simple cycles
  with a zero total displacement that contains $t_h$. It follows that there exists a
  sequence of at most $|H|(3drm)^{d}$ simple cycles
  with a zero total displacement that contains all the transitions
  $t_h$ with $h\in H$. Since the set of transitions that occurs in
  that sequence is strongly connected, Euler's Lemma shows that there
  exists a cycle $\theta$ with the same Parikh image as the sum of the
  Parikh images of the cycles occurring in the sequence. It follows that
  $|\theta|\leq r |H|(3rdm)^{d}$. Notice that
  $\Delta(\theta)=\vec{0}$ and $\theta$ is a full-state cycle. From
  $|H|\leq r(r-1)$ we are done.
\end{proof}

We deduce the following corollary.
\begin{corollary}\label{cor:smallcyle}
  The displacement of a simple cycle is the displacement of a full-state cycle with a length
  bounded by $r^2(r-1)(3drm)^{d}+r$.
\end{corollary}
\begin{proof}
  \cref{lem:t1} shows that there exists full-state cycle $\theta$ with a zero displacement with a length
  bounded by $r^2(r-1)(3drm)^{d}$. Now, just observe that any simple cycle $\theta_s$ can be inserted in $\theta$ in such a way we get a full-state cycle with the same displacement as $\theta_s$.
\end{proof}

\section{Proof of \cref{lem:pathrev}}
Now, let us prove \cref{lem:pathrev}. To do so, let us consider a structurally-reversible $I$-unfolding $G=(Q,A,T)$, and let $r\eqdef|Q|$ and $m\eqdef\normi{A}$. Notice that if $m=0$ the proof is immediate. So, we can assume that $m\geq 1$.

\medskip

Let $\vec{x},\vec{y}$ be two configurations such that the following conditions hold:
\begin{itemize}
\item $\vec{x}|_I,\vec{y}|_I\in Q$,
\item $\vec{x}(i),\vec{y}(i)\geq mr^3(3drm)^d$ for every $i\not\in I$, and
\item $\vec{y}-\vec{x}\in\vec{L}_{\vec{x}|_I,G,\vec{y}|_I}$.
\end{itemize}
Let $\pi$ be an elementary path from $\vec{x}|_I$ to $\vec{y}|_I$.

\medskip

We introduce $\vec{z}\eqdef \vec{y}-\vec{x}-\Delta(\pi)$.
Since $\vec{y}-\vec{x}\in \vec{L}_{\vec{x}|_I,G,\vec{y}|_I}$, we deduce that $\vec{z}\in\vec{L}_G$.
It follows that $\vec{z}$ is a finite sum of displacements of cycles and negation of displacements of cycles. \cref{lem:reverse} shows that the negation of the displacement of a cycle is the displacement of another cycle. It follows that $\vec{z}$ is a finite sum of displacements of cycles. As the displacement of a cycle is a finite sum of displacements of simple cycles, we deduce that $\vec{z}$ is a finite sum of displacements of simple cycles.

\medskip

\cref{cor:stei1} and \cref{cor:stei2} shows that there exists a sequence
$\vec{z}_1,\ldots,\vec{z}_k$ of displacements of simple cycles such
that $\vec{z}=\sum_{j=1}^k\vec{z}_j$, $k\leq
(1+\norm{\vec{z}})(3drm)^d$, and such that for every
$n\in\{0,\ldots,k\}$, we have:
$$\sum_{j=1}^n\vec{z}_j(i)\geq \min\{0,\vec{z}(i)\}-drm$$
\cref{cor:smallcyle} shows that for
every $1\leq j\leq k$, there exists a full-state cycle $\theta_j$ such that
$\Delta(\theta_j)=\vec{z}_j$ and
$|\theta_j|\leq r^2(r-1)(3drm)^{d}+r$.
With a rotation of $\theta_j$, we can assume without loss of generality that $\theta_j$ is a cycle on $q$. We introduce the cycle $\theta$ defined as follows:
$$\theta\eqdef\theta_1\ldots\theta_n$$

\medskip

We are going to prove that $\vec{x}\xrightarrow{\pi\theta}\vec{y}$.
To
do so, let $\delta t$ be a prefix of $\pi\theta$ where $\delta$ is a path from $p$ to a state $r$ and $t=(r,a,s)$ is a transition in $T$ and $a=(\vec{a}_-,\vec{a}_+)$ is a Petri net action. Let $i\in\{1,\ldots,d\}$ and let us prove that
$\vec{x}(i)+\Delta(\delta)(i)\geq \vec{a}_-(i)$. Observe that if $i\in I$, since $G$ is an $I$-unfolding, and $\vec{x}(i)= p(i)$,  we have $\vec{x}(i)+\Delta(\delta)(i)= r(i)$. Moreover, as $r\xrightarrow{a}s$ we deduce that $\vec{x}(i)+\Delta(\delta)(i)\geq \vec{a}_-(i)$. Now, assume that $i\not\in I$. Since $\vec{a}_-(i)\leq m$, it is sufficient to show that $\vec{x}(i)+\Delta(\pi)(i)\geq m$ in that case.

Since $\pi$ is elementary, we deduce that $|\pi|<r$. Notice that if $\delta$ is a prefix of $\pi$ then $|\delta|\leq |\pi|$. In particular $\Delta(\delta)(i)\geq -m(r-1)$. It follows that $\vec{x}(i)+\Delta(\delta)(i)\geq m$ and we are done. So, we can assume that $\delta$ is not a prefix of $\pi$. It follows that there exists $n\in\{1,\ldots,k\}$ and a prefix $\pi'$ of
$\theta_n$ such that
$\delta=\pi\theta_1\ldots\theta_{n-1}\pi'$.
Hence
$\Delta(\delta)=\Delta(\pi\pi')+\sum_{j=1}^{n-1}\vec{z}_j(i)$. Moreover,
notice that $|\Delta(\pi\pi')(i)|\leq m|\pi\pi'|\leq m(r-1)+mr^2(r-1)(3drm)^{d}+mr\leq mr^3(3drm)^d-drm-m$.
We decompose the proof
that $\vec{x}(i)+\Delta(\delta)(i)\geq m$ in two cases following that $\vec{z}(i)\leq 0$ or $\vec{z}(i)\geq 0$.
\begin{itemize}
\item Assume first that $\vec{z}(i)\geq 0$. In that case
  $\sum_{j=1}^{n-1}\vec{z}_j(i)\geq -drm$. It follows that
  $\vec{x}(i)+\Delta(\delta)(i)\geq  mr^3(3 d r m)^{d}  -drm -mr^3(3drm)^d+drm+m \geq m$. 
\item Next, assume that $\vec{z}(i)\leq 0$. In that case
  $\sum_{j=1}^{n-1}\vec{z}_j(i)\geq \vec{z}(i)-drm$. It follows that
  $\vec{x}(i)+\Delta(\delta)(i)\geq
  \vec{x}(i)+\vec{z}(i)+\Delta(\pi\pi')(i)-drm=\vec{y}(i)-\Delta(\pi\pi')(i)-drm\geq (3 d r m)^{d} -drm -m
  r^3(3drm)^{d}+drm+m\geq m$.
\end{itemize}
We have proved that $\vec{x}\xrightarrow{\pi\theta}\vec{y}$.
Now,
observe that $|\theta|\leq k(r^2(r-1)(3drm)^{d}+r)$. From
$k\leq 2\norm{\vec{z}}(3drm)^d$,
we get $|\theta|\leq \norm{\vec{y}-\vec{x}}2r^3(3drm)^{2d}$.
\cref{lem:pathrev} is proved.

\clearpage
\part{Proof of \cref{thm:witnessex}}\label{part:witnessex}\label{part:lem2}
In this part, we prove \cref{thm:witnessex}. All other results proved in this part are not used in the sequel.

\section{Extractors}\label{sec:extractor}
\newcommand{\extract}[3]{\operatorname{extract}_{#1,#3}(#2)}
The notion of extractors was first introduced in
\cite{concurjournal13}. Intuitively, extractors provide a natural way
to classify components of a vector of natural numbers
into two categories: large
ones and small ones. The notion is parameterized by a set
$I\subseteq\{1,\ldots,d\}$ that provides a way to focus only on
components in $I$.
More formally, a 
\emph{$d$-dimensional extractor} $\lambda$ is a non-decreasing sequence
$(\lambda_0\leq \cdots\leq \lambda_{d+1})$ of positive natural
numbers denoting some \emph{thresholds}. Given a $d$-dimensional extractor $\lambda$ and a set
$I\subseteq\{1,\ldots,d\}$, a \emph{$(\lambda,I)$-small set} of
a set $\vec{C}\subseteq \setN^d$ is a subset $J\subseteq I$ such that
$\vec{c}(j)<\lambda_{|J|}$ for every $j\in J$ and
$\vec{c}\in\vec{C}$. The following lemma shows that there exists a
unique maximal $(\lambda,I)$-small set w.r.t. inclusion. We denote by
$\extract{\lambda}{I}{\vec{C}}$ this set.
\begin{lemma}
  The class of $(\lambda,I)$-small sets of a set $\vec{C}\subseteq\setN^d$ is non empty and stable under union.
\end{lemma}
\begin{proof}
  We adapt the proof of \cite[Section~8]{concurjournal13}.
  Since the class contains the empty set, it is nonempty. Now, let us prove
  the stability by union by considering two $(\lambda,I)$-small sets
  $J_1$ and $J_2$ of $\vec{C}$ and let us prove that $J\eqdef J_1\cup
  J_2$ is a $(\lambda,I)$-small set of $\vec{C}$. Since $J_1,J_2\subseteq I$, we
  derive $J\subseteq I$. Let $\vec{c}\in\vec{C}$ and $j\in J$.
  If $j\in J_1$ then $\vec{c}(j)<\lambda_{|J_1|}\leq \lambda_{|J|}$
  since $|J_1|\leq |J|$. Symmetrically, if $j\in J_2$ we deduce that
  $\vec{c}(j)<\lambda_{|J_2|}\leq \lambda_{|J|}$. We have proved that
  $J$ is a $(\lambda,I)$-small set of $\vec{C}$.
\end{proof}

\begin{example}
  Let us consider the $2$-dimensional extractor $\lambda=(\lambda_0\leq \lambda_1
  \leq \lambda_2\leq \lambda_3)$ and assume that $I=\{1,2\}$ and let
  $\vec{C}=\{(m,n)\}$ with $m,n\in\setN$.
  We have:
  $$\extract{\lambda}{I}{\vec{C}}=\begin{cases}
  \{1,2\} & \text{ if }m,n<\lambda_2\\
  \emptyset & \text{ if }(m\geq \lambda_2 \wedge n\geq \lambda_1) \vee (m\geq \lambda_1
  \wedge n\geq \lambda_2)\\
  \{1\} & \text{ if } m<\lambda_1\wedge n\geq \lambda_2\\
  \{2\} & \text{ if } m\geq \lambda_2\wedge n< \lambda_1
\end{cases}$$
\end{example}

\begin{remark}
  As shown by the previous example, the values $\lambda_0$ and
  $\lambda_{d+1}$ of any $d$-dimensional extractor $\lambda$ are not used directly by our
  definitions. Those extreme values are introduced to simplify some
  notations in the sequel.
\end{remark}

The following lemma shows that components that are not in
$\extract{\lambda}{I}{\vec{C}}$ are large for at least one vector in $\vec{C}$.
\begin{lemma}\label{lem:large}
  Let $J\eqdef \extract{\lambda}{I}{\vec{C}}$. For every $i\in I\backslash J$ there exists
  $\vec{c}\in\vec{C}$ such that:
  $$\vec{c}(i)\geq \lambda_{|J|+1}$$
\end{lemma}
\begin{proof}
  Assume that for some $i\in
  I\backslash J$, we have $\vec{c}(i)<\lambda_{|J|+1}$ for every
  $\vec{c}\in\vec{C}$. Let $J'\eqdef J\cup\{i\}$ and
  observe that $J'$ is a $(\lambda,I)$-small set of $\vec{C}$. In
  fact, for every $\vec{c}\in\vec{C}$ and for every $j\in J'$, we have
  $\vec{c}(j)< \lambda_{|J|}\leq \lambda_{|J'|}$ if $j\in J$, and
  $\vec{c}(j)<\lambda_{|J|+1}=\lambda_{|J'|}$ if $j=i$. 
  We get
  a contradiction by maximality of $\extract{\lambda}{I}{\vec{C}}$. We deduce the lemma.
\end{proof}

Given a set $I\subseteq \{1,\ldots,d\}$ we define
$\extract{\lambda}{I}{e}$ for a finite word $e$ of configurations 
by $\extract{\lambda}{I}{\varepsilon}\eqdef I$, and by
$\extract{\lambda}{I}{e\vec{c}}\eqdef\extract{\lambda}{\extract{\lambda}{I}{e}}{\{\vec{c}\}}$
for every $\vec{c}\in\setN^d$ and for every finite word $e$ of configurations. Given
an infinite word $e$ of configurations, we observe that
$(\extract{\lambda}{I}{e_n})_{n\in\setN}$ where $e_n$ is the finite
prefix of $e$ of length $n$ is a non-increasing sequence of sets in
$\{1,\ldots,d\}$. It follows that this sequence is asymptotically
constant and equals to a set included in $\{1,\ldots,d\}$. We denote
$\extract{\lambda}{I}{e}$ that set. The
following lemma shows that extracting along a word of
configurations in $\vec{C}$ asymptotically
coincides with an extraction of $\vec{C}$.
\begin{lemma}\label{lem:sequential}
  Let us consider a set $I\subseteq \{1,\ldots,d\}$, an extractor
  $\lambda$, a set $\vec{C}$ of configurations, and an infinite
  word $e$ over $\vec{C}$. We have $\extract{\lambda}{I}{\vec{C}}\subseteq
  \extract{\lambda}{I}{e}$. Moreover, $\extract{\lambda}{I}{\vec{C}}=
  \extract{\lambda}{I}{e}$ if every configuration of $\vec{C}$ occurs infinitely
  often in $e$.
\end{lemma}
\begin{proof}
  We introduce $J\eqdef\extract{\lambda}{I}{\vec{C}}$, $J_\infty\eqdef
  \extract{\lambda}{I}{e}$, the prefix $e_n$ of length $n$ of $e$, and
  $J_n\eqdef\extract{\lambda}{I}{e_n}$. 

  \medskip

  Let us prove that $J\subseteq J_n$ for every $n$. 
  Since $J_0=I$ the property is proved for $n=0$. Assume that
  $J\subseteq J_{n-1}$ for some $n\geq 1$ and let us prove that
  $J\subseteq J_n$. There exists $\vec{c}\in\vec{C}$ such that
  $e_n=e_{n-1}\vec{c}$. Since $\vec{c}\in\vec{C}$, it follows that
  $\vec{c}(j)<\lambda_{|J|}$ for every $j\in J$. As $J\subseteq
  J_{n-1}$, we deduce that $J$ is a $(\lambda,J_{n-1})$-small set of
  $\{\vec{c}\}$. Since $J_n$ is the maximal set satisfying that
  property,
  we get $J\subseteq J_n$ and
  we have proved the induction. It follows that $J\subseteq J_n$ for
  every $n\in\setN$. Moreover, since
  $J_\infty=\bigcap_{n\in\setN}J_n$, we deduce the inclusion
  $J\subseteq J_\infty$.

  \medskip

  Now, assume that every $\vec{c}\in\vec{C}$ occurs in $e$ infinitely
  often. Since $(J_n)_{n\in\setN}$ is a non increasing
  sequence of $\{1,\ldots,d\}$, there exists $N$ such that
  $J_n=J_\infty$ for
  every $n\geq N$. Let $\vec{c}\in\vec{C}$. There exists $n>
  N$ such that $e_n=e_{n-1}\vec{c}$. From
  $J_{n}=\extract{\lambda}{J_{n-1}}{\{\vec{c}\}}$ and $J_{n}=J_{n-1}=J_\infty$, we
  derive $J_\infty=\extract{\lambda}{J_\infty}{\{\vec{c}\}}$. In particular
  $\vec{c}(j)<\lambda_{|J_\infty|}$ for every $j\in J_\infty$. We have
  proved that $\vec{c}(j)<\lambda_{|J_\infty|}$ for every $j\in
  J_\infty$ and for every $\vec{c}\in\vec{C}$.
  As
  $J_\infty\subseteq I$, we deduce that $J_\infty$ is a
  $(\lambda,I)$-small set of $\vec{C}$. Since $J$ is the maximal set
  satisfying that property, we
  deduce that $J_\infty\subseteq J$. It follows that $J=J_\infty$.
\end{proof}

\section{Rackoff Extraction}\label{sec:rackoff}
A \emph{$\sigma$-execution}, where $\sigma=a_1\ldots a_k$ is a word of
Petri net actions, is a non-empty word of configurations
$e=\vec{c}_0\vec{c}_1\ldots \vec{c}_k$ such that
$\vec{c}_0\xrightarrow{a_1}\vec{c}_1\cdots\xrightarrow{a_k}\vec{c}_k$.
We denote by $\src{e}$ and $\tgt{e}$ the
configurations $\vec{c}_0$ and $\vec{c}_k$ respectively. An
\emph{execution} of a PN $A$ is a $\sigma$-execution for some $\sigma\in A^*$.

An execution $e$ is said to be \emph{$I$-cyclic} for some
$I\subseteq\{1,\ldots,d\}$ if $\src{e}|_I=\tgt{e}|_I$. We say that
a word $\sigma=\vec{a}_1\ldots\vec{a}_k$ of actions in a PN $A$ is obtained from
an execution $e$ of $A$ by removing $I$-cycles where $I\subseteq
\{1,\ldots,d\}$, if there exists a decomposition of $e$ into a concatenation
$e_0\ldots e_k$ of $I$-cyclic executions $e_0,\ldots,e_k$ such that
$\tgt{e_{j-1}}\xrightarrow{a_j}\src{e_j}$ for every $1\leq j\leq k$.

\medskip

An extractor $\lambda=(\lambda_0\leq \cdots\leq \lambda_{d+1})$ is
said to be \emph{$m$-adapted} if for every $n\in \{0,\ldots,d\}$:
$$\lambda_{n+1}\geq \lambda_n+m\lambda_n^n$$

\begin{lemma}[slight extension of \cite{Rackoff78}]\label{lem:rack}
  Let $\lambda$ be an $m$-adapted extractor and $e$ be an
  execution of a PN $A\subseteq \{0,\ldots,m\}^d\times\setN^d$. Let
  $I\eqdef\extract{\lambda}{\{1,\ldots,d\}}{e}$. There
  exists a word $\sigma$ that can be obtained from $e$ by removing
  $I$-cycles such that $|\sigma|\leq d\lambda_d^d$
  and such that $\src{e}\xrightarrow{\sigma}\vec{c}$ for some
  configuration $\vec{c}$ satisfying $\vec{c}(i)=\tgt{e}(i)$ for
  every $i\in I$, and such that for every $i\not\in I$ we
  have:
  $$\vec{c}(i)\geq \lambda_{|I|+1}-m\sum_{j=1}^{|I|}\lambda_j^j$$
\end{lemma}
\begin{proof}
  The proof follows a similar approach to the original one from
  Rackoff~\cite{Rackoff78}.
  We prove the lemma by induction over $d$. Naturally, if $d=0$ the
  lemma is immediate. Assume the lemma proved for every dimension
  strictly smaller than $d\geq 1$ and let us consider an $m$-adapted extractor
  $\lambda=(\lambda_0\leq\cdots\leq \lambda_{d+1})$ and an
  $A^*$-execution $e=\vec{c}_0\ldots\vec{c}_k$ for a PN $A\subseteq
  \{0,\ldots,m\}^d\times\setN^d$. We introduce
  $J_n\eqdef\extract{\lambda}{\{1,\ldots,d\}}{\vec{c}_0\ldots\vec{c}_{n-1}}$
  for every $n\in\{0,\ldots,k+1\}$. Since
  $J_0=\{1,\ldots,d\}$, there
  exists a maximal $h\in\{0,\ldots,k+1\}$ such that
  $J_{h}=\{1,\ldots,d\}$. For every $0\leq n< h$, since
  $J_n=\{1,\ldots,d\}$, we deduce that 
  $\vec{c}_n\in \{0,\ldots,\lambda_d-1\}^d$. It follows that the
  cardinal of $\{\vec{c}_n \mid 0\leq n< h\}$ is bounded by
  $\lambda_d^d$. Without loss of generality, by removing cycles from
  the $A^*$-execution $e$, we can assume that
  $\vec{c}_0,\ldots,\vec{c}_{h-1}$ are distinct. It follows that
  $h\leq\lambda_d^d$. Notice that if $h=k+1$ we are done. So, we can
  assume that $h\leq k$. 

  \medskip
  
  Let us introduce $J\eqdef J_{h+1}$. By maximality of $h$, it follows that
  $J$ is strictly included in $\{1,\ldots,d\}$. We introduce $d'=|J|$.
  Thanks to a
  permutation of the components, we can assume without loss of
  generality that $J=\{1,\ldots,d'\}$.
  \cref{lem:large} shows that $\vec{c}_h(i)\geq \lambda_{d'+1}$ for
  every $i\in\{d'+1,\ldots,d\}$.
  We let
  $f:\setN^d\mapsto \setN^{d'}$ be the function defined by
  $f(\vec{z})=(\vec{z}(1),\ldots,\vec{z}(d'))$ for every
  $\vec{z}\in\setN^d$. 
  We also introduce the $d'$-dimensional extractor
  $\lambda'=(\lambda_0\leq \cdots \leq \lambda_{d'+1})$ and the PN
  $A'=\{(f(\vec{a}_-),f(\vec{a}_+))\mid (\vec{a}_-,\vec{a}_+)\in A\}$.
  Let us
  introduce the $(A')^*$-execution $e'=\vec{c}'_{h+1}\ldots \vec{c}'_k$
  where
  $\vec{c}'_n\eqdef f(\vec{c}_n)$, and let us
  introduce the sequence $J'_h,\ldots,J'_{k+1}$ defined by
  $J'_n\eqdef\extract{\lambda'}{\{1,\ldots,d'\}}{\vec{c}'_{h}\ldots\vec{c}'_{n-1}}$
  for every $n\in\{h+1,\ldots,k+1\}$.

  \medskip

  Let us first prove that $J'_n=J_n$ for every
  $n\in\{h+1,\ldots,k+1\}$. First of all notice that
  $J'_{h+1}\subseteq J_{h+1}$. Moreover, for every $i\in J_{h+1}$ we
  have $\vec{c}'_h(i)<\lambda'_{|J_{h+1}|}$. Hence $J_{h+1}$ is a
  $(\lambda', J'_{h+1})$-small set of $\{\vec{c}'_h\}$. By maximality
  of $J'_{h+1}$ we get $J_{h+1}\subseteq J'_{h+1}$. Hence
  $J'_{h+1}=J_{h+1}$. Assume by
  induction the property true for some $n\in\{h+1,\ldots,k\}$. Since
  $J'_{n+1}$ is a $(\lambda',J'_n)$-small set of $\{\vec{c}_n'\}$, we
  deduce that $J'_{n+1}\subseteq J'_n$ and
  $\vec{c}'_n(j)<\lambda'_{|J'_n|}$ for every $j\in J'_n$. As
  $J'_n=J_n$, and $\vec{c}'_n(j)=\vec{c}_n(j)$ for every $j\in
  \{1,\ldots,d'\}$, we deduce that $J'_n$ is a $(\lambda,J_n)$-small
  set of $\vec{c}_n$. By maximality of $J_{n+1}$, we get
  $J'_{n+1}\subseteq J_{n+1}$. Symmetrically, since $J_{n+1}$ is a $(\lambda,J_n)$-small set of $\vec{c}_n$, we
  deduce that $J_{n+1}\subseteq J_n$ and
  $\vec{c}_n(j)<\lambda_{|J_n|}$ for every $j\in J_n$. A $J'_n=J_n$,
  we deduce that $J_n$ is a $(\lambda',J_n')$-small
  set of $\vec{c}_n'$. By maximality of $J_{n+1}'$, we get
  $J_{n+1}\subseteq J'_{n+1}$. We have proved that $J'_n=J_n$ for
  every $n\in\{h+1,\ldots,k+1\}$.

  \medskip

  It follows that $J'_{k+1}=J_{k+1}=I$. 
  By induction, there exists a word $\sigma'$ that can be obtained
  from $e'$ by removing $I$-cycles such that
  $$|\sigma'|\leq \sum_{j=1}^{d'}\lambda_j^j$$
  and such that $\vec{c}_h'\xrightarrow{\sigma'}\vec{c}'$ for some
  configuration $\vec{c}'\in\setN^{d'}$ satisfying
  $\vec{c}'(i)=\vec{c}'_k(i)$ for every $i\in I$ and such that
  for every $i\in \{1,\ldots,d'\}\backslash I$ we have:
  $$\vec{c}'(i)\geq
  \lambda_{|I|+1}-m\sum_{j=0}^{|I|}\lambda_j^j$$
  
  \medskip

  Since $\sigma'$ can be obtained
  from $e'$ by removing $I$-cycles, it follow that 
  there exists a word $w$ that can be obtained from
  $\vec{c}_{h}\ldots \vec{c}_k$ by removing $I$-cycles, and such
  that $\sigma'$ is the word obtained from $w$ by applying the
  function $f$ on each action. Notice that for every
  prefix $u$ of $w$ and for every $i\in \{d'+1,\ldots,d\}$ we
  have:
  \begin{align*}
    \vec{c}_h(i)+\Delta(u)(i)& \geq \lambda_{d'+1}-m|w|\\
                             &\geq \lambda_{d'+1}-m \sum_{j=1}^{d'}\lambda_j^j\\
                             &\geq \lambda_{|I|+1}-m\sum_{j=0}^{|I|}\lambda_j^j
  \end{align*}
  The last inequality is obtained by induction by observing that
  $\lambda$ is $m$-adapted. We deduce that
  $\vec{c}_h(i)+\Delta(u)(i)\geq \lambda_0$ with the same kind of
  induction. In particular the configuration $\vec{c}\in\setN^d$
  defined by $\vec{c}(i)\eqdef\vec{c}'(i)$ if $i\in\{1,\ldots,d'\}$
  and $\vec{c}(i)\eqdef\vec{c}_{h+1}(i)+\Delta(w)(i)$ if
  $i\in\{d'+1,\ldots,d\}$ satisfies
  $\vec{c}_{h}\xrightarrow{w}\vec{c}$. Notice that
  $\vec{c}(i)=\vec{c}_k(i)$ for every $i\in I$, and for every
  $i\not\in I$, we have:
  $$\vec{c}(i)\geq
  \lambda_{|I|+1}-m\sum_{j=0}^{|I|}\lambda_j^j$$
  Let us introduce $\sigma\eqdef \vec{a}_1\ldots\vec{a}_{h}w$ where
  $\vec{a}_n\eqdef\vec{c}_n-\vec{c}_{n-1}$ for every
  $n\in\{1,\ldots,h\}$. Observe that
  $\vec{c}_0\xrightarrow{\sigma}\vec{c}$ and moreover we have:
  $$|\sigma|\leq h+\sum_{j=1}^{d'}\lambda_j^j\leq \sum_{j=1}^d\lambda_j^j\leq d\lambda_d^d$$
  We have proved the induction.
\end{proof}

\section{From SCCC to small unfoldings}
We associate with an extractor $\lambda$ and a SCCC $\vec{C}$ of a PN
$A$, the set $I\eqdef\extract{\lambda}{\{1,\ldots,d\}}{\vec{C}}$. Notice that $Q\eqdef\vec{C}|_I$ is finite since this set is included in $\{q\in\setN^I \mid \normi{q}<\lambda_{|I|}\}$. It follows that $G_{\vec{C},I}$ is a structurally-reversible $I$-unfolding. We show in this section that for every $\vec{c}\in\vec{C}$ there exists a kind of partial pumping pairs for $(\vec{c}, G_{\vec{C},I})$.

\medskip

Let us recall that an \emph{infinite execution} $e$ of $A$ is an infinite word of configurations such
that every finite non-empty prefix is an execution of $A$. Let us prove the following technical lemma.
\begin{lemma}\label{lem:infiniteexe}
  If $\vec{C}$ is not reduced to a singleton, there exists an
  infinite execution $e\in\vec{C}^\omega$ of $A$
  such that every configuration of $\vec{C}$ occurs
  infinitely often in $e$.
\end{lemma}
\begin{proof}
  Since $\vec{C}$ is countable, there exists an infinite sequence
  $(\vec{c}_n)_{n\in\setN}$ such that $\vec{C}=\{\vec{c}_n\mid
  n\in\setN\}$. Moreover, by replacing that sequence by the sequence
  $s_0,s_1,\ldots$ where $s_n\eqdef \vec{c}_0,\ldots, \vec{c}_n$, we can assume
  without loss of generality that every configuration of $\vec{C}$
  occurs infinitely often in the sequence
  $(\vec{c}_n)_{n\in\setN}$. Since $\vec{C}$ is a SCCC, for every
  positive natural number $n$, there exists an $A^*$-execution from
  $\vec{c}_{n-1}$ to $\vec{c}_n$ of the form $e_n\vec{c}_n$. Let us
  introduce the word $e\eqdef e_1e_2\ldots$. Notice that since
  $\vec{C}$ is not reduced to a singleton, the word $e$ is
  infinite. Moreover, notice that $e$ is an infinite execution
  satisfying the lemma.
\end{proof}

Now, assume that $\lambda$ is $m$-adapted for some positive natural
number $m$ (this notion is introduced in the previous section).

\begin{lemma}\label{lem:pumpupG}
  If $A\subseteq\{0,\ldots,m\}^d\times\setN^d$, for every $\vec{c}\in\vec{C}$, there exists a cycle $\alpha$ in $G$ on $\vec{c}|_I$ such that $|\alpha|\leq d\lambda_d^d$ and a
  configuration $\vec{c}^+$ such that
  $\vec{c}\xrightarrow{\alpha}\vec{c}^+$ and such that $\vec{c}^+(i)\geq
  \lambda_{|I|+1}-m\sum_{j=1}^{|I|}\lambda_j^j$ for every $i\not\in I$.
\end{lemma}
\begin{proof}
  Observe that if $\vec{C}$ is reduced to a singleton, the lemma is
  trivial with $\alpha\eqdef\varepsilon$. So, we can assume that $\vec{C}$
  is not a singleton.
  \cref{lem:infiniteexe} shows that there exists an infinite execution $e=\vec{c}_0\vec{c}_1\ldots$ of
  configurations in $\vec{C}$ such that every configuration of $\vec{C}$ occurs
  infinitely often.
  Without loss of generality, by replacing $e$ by a suffix of $e$
  we can assume that
  $\vec{c}=\vec{c}_0$. \cref{lem:sequential}
  shows that $\extract{\lambda}{\{1,\ldots,d\}}{e}=I$. It follows that
  there exists $N\in\setN$ such that for every $n\geq N$ the prefix
  $e_n$ of $e$ of length $n$ satisfies
  $\extract{\lambda}{\{1,\ldots,d\}}{e_n}=I$. Since $\vec{c}$ occurs
  infinitely often in $e$, there exists $n\geq N$ such that $\vec{c}$
  is the last configuration of $e_n$. 
  \cref{lem:rack} shows
  that there
  exists a word $u$ that can be obtained from $e_n$ by removing
  $I$-cycles such that $|u|\leq d\lambda_d^d$
  and such that $\vec{c}\xrightarrow{u}\vec{c}^+$ for some
  configuration $\vec{c}^+$ satisfying $\vec{c}^+(i)\geq
  \lambda_{|I|+1}-m\sum_{j=1}^{|I|}\lambda_j^j$ for every $i\not\in I$.
  Since $u$ can be obtained from $e_n$ by removing
  $I$-cycles, it follows that $u$ is the label of a cycle $\alpha$ on
  $\vec{c}|_I$ in $G$.
\end{proof}

Symmetrically, we deduce a similar backward property.
\begin{lemma}\label{lem:pumpdownG}
  If $A\subseteq \setN^d\times \{0,\ldots,m\}^d$, 
  for every $\vec{c}\in\vec{C}$, there exists a cycle $\beta$ in $G$ on $\vec{c}|_I$ such that $|\beta|\leq d\lambda_d^d$ and a
  configuration $\vec{c}^-$ such that
  $\vec{c}^-\xrightarrow{\beta}\vec{c}$, and
  such that for every $i\not\in I$:
  $\vec{c}^-(i)\geq \lambda_{|I|+1}-m\sum_{j=1}^{|I|}\lambda_j^j$.
\end{lemma}
\begin{proof}
  Let us introduce the PN $A'\eqdef\{(\vec{a}_+,\vec{a}_-)\mid
  (\vec{a}_-,\vec{a}_+)\in A\}$. Observe that $\vec{C}$ is a SCCC of
  $A'$. Let $G'$ be the $I$-unfolding associated to the extractor $\lambda$
  and the SCCC $\vec{C}$ of $A'$.
  \cref{lem:pumpup} shows that there exists a cycle in $G'$ on $\vec{c}|_I$ labeled
  by a word $u$ such that $|u|\leq d\lambda_d^d$
  and a
  configuration $\vec{c}^-$ such that
  $\vec{c}\xrightarrow{u}\vec{c}^-$, and
  such that 
  $\vec{c}^-(i)\geq \lambda_{|I|+1}-m\sum_{j=1}^{|I|}\lambda_j^j$ for
  every $i\not\in I$.
  Assume that $u=a_1'\ldots a_n'$ with $a_j'=(\vec{x}_j,\vec{y}_j)$
  and let
  $v\eqdef a_n\ldots a_1$ with $a_j\eqdef (\vec{y}_j,\vec{x}_j)$.
  Observe that since
  $u$ is the label of a cycle on $\vec{c}|_I$ in $G'$, then $v$ is the label of a
  cycle on $\vec{c}|_I$ in $G$. Moreover, from
  $\vec{c}\xrightarrow{v}\vec{c}^-$ we derive
  $\vec{c}^-\xrightarrow{v}\vec{c}$. We have proved the lemma.
\end{proof}

\section{Proof of \cref{thm:witnessex}}\label{sec:witnessex}
In this section, we prove \cref{thm:witnessex}. We consider a PN
$A$ and a SCCC $\vec{C}$ of $A$. We introduce $m\eqdef \normi{A}$. If $m=0$, the proof is immediate. So, we can assume that $m\geq 1$.

\medskip

We introduce the extractor $\lambda$ satisfying $\lambda_0=1$, and for
every $n\in\{0,\ldots,d\}$:
$$\lambda_{n+1}\eqdef m\sum_{j=1}^n\lambda_j^j+m\lambda_n^{3n}(3d\lambda_n^nm)^{d}$$
Observe that $\lambda$ is $m$-adapted.

We introduce $b\eqdef \lambda_d$, the set $I\eqdef\extract{\lambda}{\{1,\ldots,d\}}{\vec{C}}$, and the
structurally-reversible $I$-unfolding $G=(Q,A,T)$
associated to $\vec{C}$ (defined in the previous section), $\lambda$ and $A$.
Notice that $Q\subseteq \{q\in\setN^I \mid \normi{q}<\lambda_{|I|}\}$.
Denoting by $r\eqdef |Q|$, we deduce that $r\leq \lambda_{|I|}^{|I|}$. 

\medskip

Observe that for every $\vec{x},\vec{y}\in \vec{C}$, we have $\vec{y}-\vec{x}\in L_{\vec{x}|_I,G,\vec{y}|_I}$. Moreover  $Q\subseteq \{q\in\setN^I \mid \normi{q}<b\}$ since $\lambda_{|I|}\leq b$. The following lemma provides a bound on $b$.
\begin{lemma}\label{lem:t42}
  We have $b\leq (3dm)^{(d+2)^{2d+1}}$.
\end{lemma}
\begin{proof}
  Notice that $\lambda_{n+1}= m\sum_{j=1}^n\lambda_j^j+m\lambda_n^{3n}(3d\lambda_n^nm)^{d}\leq mn\lambda_n^n+m\lambda_n^{3n}(3d\lambda_n^nm)^{d}\leq md\lambda_n^{3d}+m\lambda_n^{3d}(3d\lambda_n m)^{d^2}\leq (3d\lambda_n m)^{d^2+3d}$. By induction on $n$, we deduce that $\lambda_n\leq (3dm)^{n(d^2+3d)^n}$. In particular, we have $b\leq (3dm)^s$ where $s\eqdef d(d^2+3d)^d\leq (d+2)^{2d+1}$.
\end{proof}

\medskip

\begin{lemma}\label{lem:pumpup}
  We have $\vec{c}\in \vec{U}_{\vec{c}|_I,G}$ for every $\vec{c}\in\vec{C}$.
\end{lemma}
\begin{proof}
  \cref{lem:pumpupG} and \cref{lem:pumpdownG} show that there exist $u,v\in A^*$ that label cycles $\alpha,\beta$ in $G$ on $\vec{c}|_I$ such that $|u|,|v|\leq db^d$, and two configurations $\vec{c}^-,\vec{c}^+$ such that $\vec{c}^-\xrightarrow{u}\vec{c}\xrightarrow{v}\vec{c}^+$, and such that $\vec{c}^-(i),\vec{c}^+(i)\geq \lambda_{|I|+1}-m\sum_{j=1}^{|I|}\lambda_j^j$ for every $i\not\in I$. Notice that $\lambda_{|I|+1}-m\sum_{j=1}^{|I|}\lambda_j^j=m (\lambda_{|I|}^{|I|})^3(3d \lambda_{|I|}^{|I|}m)^{d}\geq mr^3(3drm)^d$. It follows that $(u,v)$ is a pumping pair for $(\vec{c}|_I,G)$.
\end{proof}

\clearpage
\part{Conclusion}\label{part:conclusion}

This paper provides a way for computing a quantifier free Presburger formula $\phi_A(\vec{x},\vec{y})$ encoding the mutual reachability relation of a Petri net $A$ between two configurations $\vec{x},\vec{y}$. We also provided in \cref{thm:bot} a Presburger formula encoding the set of bottom configurations. This formula introduces quantified sub-formulas of the form $\forall \vec{v}\in\crochet{\gamma}\phi(\vec{c}+\vec{v})$ or equivalently $\exists \vec{v}\in \crochet{\gamma}\neg\phi(\vec{c}+\vec{v})$ where $\gamma$ is a representation of a lattice, and $\phi$ is a threshold formula. In order to obtain a quantifier-free formula, by putting $\neg \phi$ in disjunctive normal form, it is sufficient to provide a way to encode in the quantifier-free fragment of the Presburger arithmetic the following set:
$$(I_1\times\cdots\times I_d)+\setZ\vec{p}_1+\cdots+\setZ\vec{p}_k$$
where $I_1,\ldots,I_d$ are intervals of $\setZ$ and $\vec{p}_1,\ldots,\vec{p}_k$ are vectors in $\setZ^d$. We left this problem open.

\bigskip

The author thanks Petr Jan\v{c}ar for discussions motivating this work.



\bibliography{main.bib}
\end{document}